\documentclass[12pt]{article}
\usepackage{cite}
\usepackage{graphicx}
\usepackage[labelformat=simple]{subcaption}

\usepackage{amsfonts}
\usepackage{amssymb}
\usepackage{amsmath}
\usepackage{color}
\usepackage{clay}
\usepackage{hyperref}
\hypersetup{colorlinks=true}
\hypersetup{linkcolor=black}
\hypersetup{citecolor=black}
\hypersetup{urlcolor=black}
\hypersetup{linktocpage=true}
\usepackage{setspace}
\usepackage{verbatim}
\numberwithin{equation}{section}
\usepackage{stmaryrd}
\usepackage{slashed}
\usepackage{centernot}
\usepackage{enumerate}

\usepackage[boxsize=0.5em,aligntableaux=center]{ytableau}
\usepackage{amsthm}

\theoremstyle{plain}
\newtheorem{theorem}{Theorem}

\newtheorem{conjecture}{Conjecture}

\theoremstyle{definition}
\newtheorem{definition}{Definition}

\RequirePackage{ifpdf}

\newcommand{\T}{\top}
\DeclareMathOperator{\Tr}{Tr}

\newcommand{\dd}{\mathrm{d}}
\newcommand{\DD}{\mathrm{D}}
\renewcommand{\Re}{\ensuremath{\operatorname{Re}}}
\renewcommand{\Im}{\ensuremath{\operatorname{Im}}}

\newcommand{\Gam}[1]{\ensuremath{\Gamma{\left({#1}\right)}}}
\newcommand{\GamF}[2]{\ensuremath{\Gamma_{\!\frac{#1}{#2}}}}
\newcommand{\floor}[1]{\ensuremath{\left\lfloor{#1}\right\rfloor}}
\newcommand{\ceil}[1]{\ensuremath{\left\lceil{#1}\right\rceil}}
\newcommand{\dbracket}[1]{\ensuremath{\llbracket {#1} \rrbracket}}

\newcommand{\bZ}{\mathbb{Z}}

\newcommand{\be}{\begin{equation}}
\newcommand{\ee}{\end{equation}}

\renewcommand{\ge}{\geqslant}
\renewcommand{\le}{\leqslant}

\newcommand{\pfq}[2]{\ensuremath{{}_{#1}F_{#2}}}

\newmuskip\pFqskip
\mathchardef\pFcomma=\mathcode`, 

\newcommand*\pFq[5]{%
  \begingroup
  \begingroup\lccode`~=`,
    \lowercase{\endgroup\def~}{\pFcomma\mkern6mu}%
  \mathcode`,=\string"8000
  {}_{#1}F_{#2}\!\biggl[\genfrac..{0pt}{}{#3}{#4};#5\biggr]%
  \endgroup
}

\begin{document}

\renewcommand*{\thefootnote}{\alph{footnote}}

\begin{titlepage}

\begin{center}

~\\[1cm]

{\Huge Orbifolds and Exact Solutions \\[5pt] of Strongly-Coupled Matrix Models}

~\\[1cm]

Clay C\'{o}rdova,$^1$\footnote{e-mail: {\tt claycordova@ias.edu}} Ben Heidenreich,$^{2}$\footnote{e-mail: {\tt bheidenreich@perimeterinstitute.ca}} Alexandr Popolitov,$^{3,4,5}$\footnote{e-mail: {\tt popolit@gmail.com}} and Shamil Shakirov$^{4,6}$\footnote{e-mail: {\tt shakirov@g.harvard.edu}}

~\\[0.1cm]

$^1$ {\it School of Natural Sciences, Institute for Advanced Study, Princeton, NJ 08540, USA}

$^2$ {\it Perimeter Institute for Theoretical Physics, Waterloo, Ontario, Canada N2L 2Y5}

$^3$ {\it Institute for Theoretical and Experimental Physics, Moscow 117218, Russia}

$^4$ {\it Institute for Information Transmission Problems, Moscow 127994, Russia}

$^5$ {\it Korteweg-de Vries Institute for Mathematics, University of Amsterdam, \\ P.O. Box 94248, 1090 GE Amsterdam, The Netherlands}

$^6$ {\it Society of Fellows, Harvard University, Cambridge, MA 20138, USA}

~\\[0.8cm]

\end{center}

\begin{center}
\textbf{Abstract}
\end{center}

We find an exact solution to strongly-coupled matrix models with a single-trace monomial potential.  Our solution yields closed form expressions for the partition function as well as averages of Schur functions.   The results are fully factorized into a product of terms linear in the rank of the matrix and the parameters of the model.  We extend our formulas to include both logarthmic and finite-difference deformations, thereby generalizing the celebrated Selberg and Kadell integrals.  We conjecture a formula for correlators of two Schur functions in these models, and explain how our results follow from a general orbifold-like procedure that can be applied to any one-matrix model with a single-trace potential.
\vfill

\begin{flushleft}
\today
\end{flushleft}

\end{titlepage}

\renewcommand*{\thefootnote}{\arabic{footnote}}

\setcounter{tocdepth}{2}
\tableofcontents

\section{Introduction}

In this paper we study a class of matrix models, those with single-trace potential of monomial form
\begin{equation}
 S[X]=\mathrm{Tr}(X^{r})~,
 \end{equation}
 and generalizations thereof.  When $r=2$ the model is quadratic and free, but for $r>2$ the models we study are interacting and can be viewed as the infinite coupling limit of more familiar Gaussian plus interaction potentials.

 We demonstrate that any such monomial matrix model is exactly solvable and provide a completely factorized form of the correlators.  The solution depends on a non-perturbative choice of contour of integration in the space of matrices, which introduces an additional hidden integral parameter $0\leq a<r$ into the model.  A taste of the type of formulas that we provide is the following expression for the partition function
\begin{equation}
Z_N^{(r,a)} = \frac{\delta_{r,a}(N)}{(2\pi)^N} \prod_{i=0}^{N-1}  \Gam{\floor{\frac{i}{r}}+1} \Gam{\floor{\frac{i-a}{r}}+\frac{a}{r}+1} ~,
\end{equation}
where $N$ is the rank of the matrix and $\delta_{r,a}(N) = 0,\pm 1$, depending on $N$, see~(\ref{eqn:deltaraN}).
In~\S\ref{sec:exactsolutions} we provide a similar formula for the expectation value of Schur polynomial insertions $s_{\lambda}(X)$ where $\lambda$ is a Young diagram of at most $N$ rows:
\begin{equation}  \label{eqn:schurpreview}
\langle s_{\lambda} (X) \rangle = \frac{\delta_r(\lambda)}{r^{|\lambda|/r}} \prod_{x \in \lambda} \frac{\llbracket N + c_\lambda(x) \rrbracket_{r,0}\, \llbracket N + c_\lambda(x) \rrbracket_{r,a}}{\llbracket h_\lambda(x) \rrbracket_{r,0}} \,.
\end{equation}
Here $c_{\lambda}$ and $h_{\lambda}$ are respectively the contents and hook length of the box $x\in\lambda$, $\delta_{r}(\lambda)= 0, \pm1$ depending on $\lambda$ (see Appendix~\ref{app:schurdivisibility}), and
\begin{equation}
\llbracket n \rrbracket_{r, a} = \begin{cases} n & n \equiv a \mod r \\ 1 & \text{otherwise} \end{cases} \,.
\end{equation}
These formulas---and their logarithmic and $q$-deformed (finite difference calculus) analogs that we obtain in \S\ref{sec:logmodels} and \S\ref{subsec:qanalogs}---are similar to the celebrated Selberg and Kadell integrals~\cite{Selberg:1941,Selberg:1944,Kadell:1988,Kadell:1997}, and reduce to them when $r=1$. To illustrate their simple and explicit nature, we show an example of applying~(\ref{eqn:schurpreview}) in Figure~\ref{fig:3foldCorr}.
\begin{figure}
  \centering
  \includegraphics[width=0.75\textwidth]{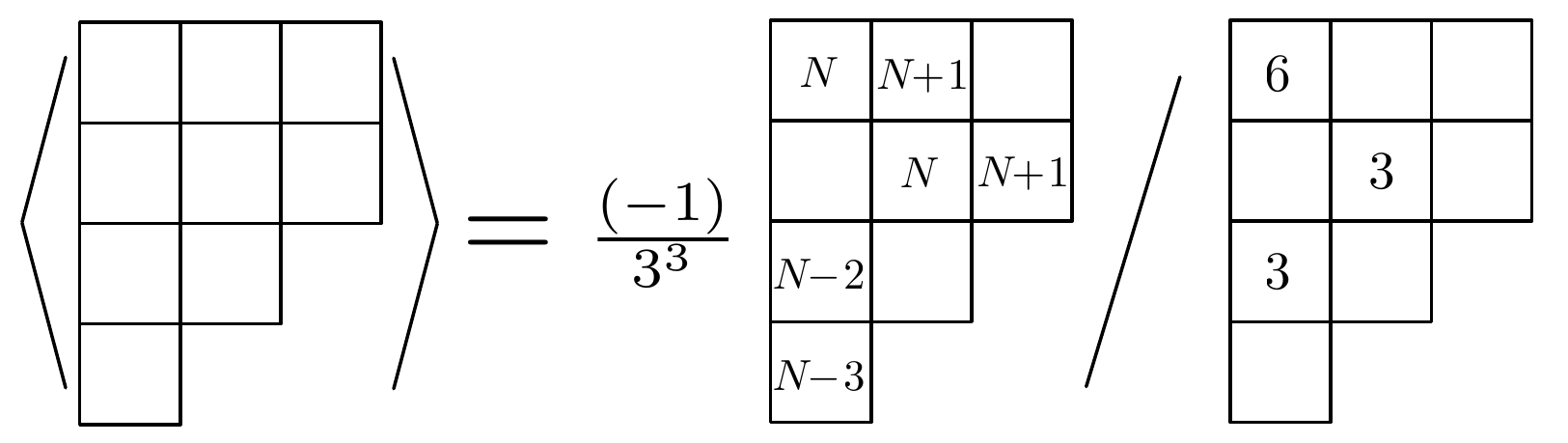}
  \caption{Evaluating a Schur polynomial average using~(\ref{eqn:schurpreview}) for the pure cubic potential $S[X] = \Tr (X^3)$ with $a=1$ and $N \equiv 0 \pmod 3$. In this example, we find $\langle s_{3,3,2,1}(X) \rangle = -\frac{N^2 (N+1)^2 (N-2) (N-3)}{2\times3^6}$ where the $3$-signature $\delta_3(\lambda)=-1$ can be found using one of several methods discussed in Appendix~\ref{app:schurdivisibility}.\label{fig:3foldCorr}}
\end{figure}

In \S\ref{sec:genconstruction} we show that our results are related to a general orbifold-like structure in matrix models, where the partition function and Schur polynomial averages for the potential $W(X^r)$ factor into $r$-fold products of those for the potential $W(X)$. This factorization is related to a combinatoric identity involving the Vandermonde, and occurs in an natural basis of complex integration contours. (For other integration contours the partition function and insertions can be expressed as a sum over factorized components.)

We emphasize that these results go far beyond the observation that one-matrix models with single-trace potentials are solvable, e.g., by the method of orthogonal polynomials. Generically, even after determining the moments of the eigenvalue potential, computing the partition function (with or without insertions) requires the evaluation of an $N\times N$ determinant, by diagonalization (as in the method of orthogonal polynomials) or otherwise. By contrast, in the examples we study these determinants can be evaluated in closed form for arbitrary $N$. One illustration of the simplicity of our results is the fact that correlation functions of $N$-independent operators are rational functions of $N$ (once discrete data such as $a$ and $(N \bmod r)$ are fixed). This suggests that these models fall into a class which is some matrix-model analog of exactly solvable quantum field theories, such as integrable systems, rational CFTs, and Liouville theory. (However, we know no precise definition of such a class.)

Our results have a variety of possible applications.  As suggested above, the simple potential
\begin{equation}
S[X]=\mathrm{Tr}(X^{2})+\lambda \mathrm{Tr}(X^{r})~,
\end{equation}
can be viewed as a toy model of variety of interacting quantum systems, where $\lambda$ controls the strength of the interactions.  When $\lambda$ is small perturbation theory can be applied, but the resulting perturbative expansion is not convergent, and advanced techniques are required to make sense of it (see, e.g.,~\cite{Dunne:2012ae} for a recent discussion).  On the other hand, when the coupling $\lambda$ is large, it is more appropriate to view the Gaussian potential $\Tr(X^2)$ as a perturbation around the monomial model that we solve exactly. As we demonstrate in \S\ref{expans}, the resulting strong coupling expansion converges for any $\lambda \ne 0$, and provides a different perspective on the interacting system.

Apart from their use as toy models, our monomial matrix models may also find applications in calculations of certain observables in ordinary quantum field theory.  For instance the case $r=1$ (the Selberg integral) is related to the integral expression for the superconformal index of a class of four-dimensional $\mathcal{N}=1$ gauge theories~\cite{vanDiejen:2001,Forrester:2007}.\footnote{Specifically, those with gauge group $Sp(2N)$, one chiral multiplet in the antisymmetric tensor representation, six in the fundamental representation, and no superpotential. These theories confine without chiral symmetry breaking~\cite{Cho:1996bi,Csaki:1996eu}.}  The fact that the integral can be evaluated in closed form is a reflection of $s$-confinement~\cite{Seiberg:1994bz,Seiberg:1994pq,Csaki:1996sm,Csaki:1996zb}: this theory has an alternative infrared description in terms of free fields.  The $r>1$ models we study here may be related to the superconformal index of the same class of $s$-confining gauge theories on a Lens space~\cite{Benini:2011nc,Razamat:2013opa,Spiridonov:2016uae}.  We comment further on this in \S\ref{subsec:index}.

Related matrix models ($r=3$) have also appeared in the study of five-dimensional gauge theory on the five-sphere~\cite{Minahan:2014hwa}.  In that case the quadratic piece of the model controls the Yang-Mills term and the cubic descends from the Chern-Simons interaction.  The pure monomial model describes the strongly-coupled pure Chern-Simons theory that arises at infinite Yang-Mills coupling.

Finally, another possible avenue of application for our results is described in \S\ref{subsec:instanton} concerns a generalization of the AGT conjecture \cite{Alday:2009aq,Wyllard:2009hg}.  In~\cite{Kimura:2011zf} the partition function of four-dimensional gauge theory on an orbifold $\mathbb{C}^{2}/\mathbb{Z}_{r}$ was studied resulting in matrix models similar to those considered here.  It would be interesting to determine the precise connection, and to understand if our results may be used to compute three-point functions in para-Liouville theory (see, e.g., \cite{Bonelli:2009zp,Belavin:2011pp,Nishioka:2011jk,Bonelli:2011jx,Bonelli:2011kv}), thus generalizing~\cite{Dotsenko:1984nm}.\footnote{There is also a possible connection between our results and the $(r+1)$-point function of ordinary Louiville theory for a special arrangements of the points.} We leave this as a potential direction for future research.  

\section{Matrix Model Review}

In this section we review general properties of random matrix models which are pertinent to our results.  For more comprehensive reviews, see, e.g., \cite{DiFrancesco:1993cyw,Marino:2004eq}. We focus on the relation between perturbative and non-perturbative approaches, the role of the integration contour in the latter and the role of reflection positivity.

A matrix model is an average over random matrices, of the schematic form:
\be \label{eqn:Zintegral}
Z = \frac{1}{\mathrm{Vol\ } G} \int \DD X\, e^{-S[X]} \,,
\ee
where $S[X]$ is an action that depends on one or more matrices $X$, $\DD X$ is some measure for integration over these matrices, and $G$ is a gauge group, i.e.\ a symmetry of the action whose singlet sector defines the set of observables. Observables are defined by inserting $X$-dependent operators into the partition function:
\be \label{eqn:ZOintegral}
Z[\mathcal{O}] \equiv \frac{1}{\mathrm{Vol\ } G} \int \DD X\, \mathcal{O}[X]\, e^{-S[X]} \,,
\ee
whereas normalized expectation values are obtained by dividing by the partition function
\be \label{eqn:ZOintegral-normalized}
\langle \mathcal{O} \rangle \equiv \frac{Z[\mathcal{O}]}{Z} \,.
\ee
In the case where $G$ is non-trivial, only $G$-invariant operators are permissible.

There is an obvious analogy between matrix models and quantum field theories; in essence, a matrix model is a quantum field theory in zero dimensions. We pursue this analogy in more detail below, as it will provide an interesting context for our main results.

As an example, Hermitian one-matrix models are defined by the measure
\be \label{eqn:Hermitianmeasure}
\DD X = \prod_{i=1}^N \dd X_{i i} \prod_{i<j}^N \dd X_{i j}\, \dd X_{i j}^\ast
\ee
where $X$ is a Hermitian $N\times N$ matrix and the independent real components are integrated from $-\infty$ to $\infty$. The measure is invariant under $U(N)$ transformations
\be
X \to U^\ast X U \,,
\ee
where $U^\ast$ denotes the Hermitian conjugate of $U$,
so it is natural to take $G = U(N)$.

Henceforward, we focus on such $U(N)$ Hermitian matrix models with single-trace potentials of the form $S[X] = \Tr W(X)$, where $W$ is any analytic function. Expanding about a critical point of $W$, we obtain
\be \label{eqn:interactingmodel}
S[X] = \frac{1}{2} \Tr X^2 +\sum_{p=3}^\infty \lambda_p \Tr X^p \,.
\ee
A Gaussian matrix model has $S[X] = \frac{1}{2} \Tr X^2$, and is free in the sense that $\Tr X^2 = \sum_{i, j} X_{i j} X_{j i} = \sum_{i, j} |X_{i j}|^2$, so that the integral factors into one-dimensional integrals. Models with $\lambda_p \ne 0$ for $p>2$ are interacting, in that the integral no longer factors in this way. Two examples of interacting models that we will use frequently are the cubic model, $S[X] = \frac{1}{2} \Tr X^2 + \lambda_3 \Tr X^3$, and the quartic model $S[X] = \frac{1}{2} \Tr X^2 +\lambda_4 \Tr X^4$.  Interacting theories can be studied perturbatively by splitting $S =  \frac{1}{2} \Tr X^2 + S_{\rm int}$ and expanding $e^{-S_{\rm int}}$ in a formal power series in the coupling constants.

\subsection{The loop equations}

An alternative approach to these matrix models is to systematically exploit integration by parts identities (a.k.a.\ Ward identities).  The resulting formulas are known as \emph{loop equations} (see e.g.~\cite{DiFrancesco:1993cyw}), and we review them below.

Let $M[X] = M_0(\Tr X^i) + M_1(\Tr X^i) X + M_2(\Tr X^i) X^2+\ldots$ be a polynomial function of $X$ and its traces, and introduce the matrix differential operator
\begin{equation}
(\partial_{X})_{i,j}=\frac{\partial}{\partial (X)_{j,i}}~.
\end{equation}
Now consider the total derivative:
\be \label{eqn:loopboundary}
0 = \frac{1}{\mathrm{Vol\ } G} \int \DD X\, \Tr \partial_X (M[X] e^{-S[X]}) \,,
\ee
where we assume that the corresponding boundary term vanishes. We then obtain:
\be \label{eqn:loopequation0}
  0 =
  Z\! \left[ \Tr(\partial_X M)  - \Tr (M \partial_X) S \right] \,.
\ee
In particular, taking $M = \mathcal{O} X^{k+1}$ for $k\ge -1$ for some gauge invariant operator $\mathcal{O} = \mathcal{O}(\Tr X^i)$ and using $\partial_X X^{n+1} = \sum_{i=0}^n (\Tr X^i) X^{n-i}$ and $\partial_X \Tr X^{n+1} = X^n$, we find:
\be \label{eqn:loopequation}
  0 =
  Z\! \left[  \mathcal{O} \sum_{i = 0}^k \Tr X^i \Tr X^{k - i} + \Tr (X^{k + 1} \partial_X) \mathcal{O} - \mathcal{O} \Tr (X^{k + 1} \partial_X) S \right] \,,
\ee
for $k \ge -1$.

These are the loop equations for the one-matrix model~(\ref{eqn:Zintegral}). Besides the action $S[X]$, the only information about the matrix integral that (\ref{eqn:loopequation}) encodes is the vanishing of the boundary term~(\ref{eqn:loopboundary}). Nonetheless, the loop equations are a powerful tool for solving the model. For instance, the Gaussian matrix model, $S = \frac{1}{2} \Tr X^2$, has the loop equations
\be \label{eqn:Gloopequation}
  \left\langle  \mathcal{O} \Tr X^{k + 2} \right\rangle =
  \sum_{i = 0}^k  \left\langle  \mathcal{O} \Tr X^i \Tr X^{k - i}\right\rangle + \left\langle \Tr (X^{k + 1} \partial_X) \mathcal{O} \right\rangle \,.
\ee
Suppose that $\mathcal{O}$ is a polynomial in $X$, and split it into monomials. Notice that the total power of $X$ on the left-hand side (LHS) of the equation is two greater than on the right-hand side (RHS). We can compute any polynomial correlator by iteratively replacing terms of the form $\langle  \mathcal{O} \Tr X^{k+2} \rangle$, $k \ge -1$, with the RHS of~(\ref{eqn:Gloopequation}). Since every non-trivial gauge-invariant monomial in $X$ takes this form, and the overall power of $X$ is strictly decreasing, this procedure reduces every polynomial correlator to $\langle 1 \rangle = 1$, allowing all such correlators to be computed.

Thus, the loop equations provide a nearly-complete solution to the Gaussian matrix model, fixing everything but the partition function itself (which requires a direct evaluation of the integrals). The loop equations also provide a perturbative solution to interacting models. The action~(\ref{eqn:interactingmodel}) gives the loop equations:
\be
 \left\langle  \mathcal{O} \Tr X^{k + 2} \right\rangle =
  \sum_{i = 0}^k  \left\langle  \mathcal{O} \Tr X^i \Tr X^{k - i}\right\rangle + \left\langle \Tr (X^{k + 1} \partial_X) \mathcal{O} \right\rangle - \sum_{p>2} p \lambda_p \left\langle  \mathcal{O} \Tr X^{k + p} \right\rangle \,.
\ee
Because of the last term, the largest power of $X$ on the RHS is now larger than on the LHS, and we cannot solve the model exactly using the same method as before. However, if we truncate at some fixed order $\lambda_p^{k_p}$ in perturbation theory with $\sum_p k_p$ finite, then the power of $X$ increases in at most $\sum_p k_p$ steps by at most $\sum_p (p-2) k_p$ in total, and the same iterative procedure as for the free theory terminates in a finite number of steps.

Thus, the loop equations fix the correlators of both the free theory and perturbatively in interacting theories.

\subsection{Contour dependence} \label{sec:contourdependence}

It is interesting to ask whether the loop equations likewise solve interacting theories \emph{non-perturbatively}. As an example, we consider the cubic matrix model $S = \frac{1}{2} \Tr X^2 + \lambda_3 \Tr X^3$. We rewrite the loop equations as
\be
3 \lambda_3 \left\langle \mathcal{O} \Tr X^{k + 3} \right\rangle = \sum_{i = 0}^k  \left\langle  \mathcal{O} \Tr X^i \Tr X^{k - i}\right\rangle + \left\langle \Tr (X^{k + 1} \partial_X) \mathcal{O} \right\rangle - \left\langle  \mathcal{O} \Tr X^{k + 2} \right\rangle \,,
\ee
so that the LHS has a power of $X$ one greater than the RHS. We can then iteratively apply these equations to simplify any polynomial insertion, just as for the Gaussian theory. As before, this procedure terminates in a finite number of steps due to the strictly decreasing maximum power of $X$. However, unlike in the Gaussian case, not every monomial insertion is of the form $\langle \mathcal{O} \Tr X^{k+3} \rangle$ for $k\ge -1$, hence the best we can do is to express every correlator in terms of the unknown averages $\langle (\Tr X)^p \rangle$ for $p > 0$, where $\langle (\Tr X)^0 \rangle = 1$.

This is not the full story, because for finite-dimensional $N\times N$ matrices there are operator equations\footnote{Recall that an operator equation $A = B$ holds if and only if $\langle A \mathcal{O} \rangle = \langle B \mathcal{O} \rangle$ for any operator $\mathcal{O}$.} relating $\Tr X^P$, $P>N$, to polynomials in $\Tr X^p$, $p\le N$. These trace relations exist because gauge-invariant polynomials in $X$ only depend on the $N$ eigenvalues of $X$, so there are only $N$ independent gauge-invariant operators. For instance,
\begin{align} \label{eqn:opeqnN}
\Tr X^{N+1} = \frac{(-1)^{N+1}}{N!} (\Tr X)^{N+1} + \ldots \,,
\end{align}
where the omitted terms involve at least one factor of $\Tr X^p$, $2\le p\le N$. Applying the loop equations to~(\ref{eqn:opeqnN}), we can express $\langle (\Tr X)^{N+1} \rangle$ in terms of $\langle (\Tr X)^p \rangle$, $p\le N$. Multiplying~(\ref{eqn:opeqnN}) by $(\Tr X)^k$ and applying the loop equations once more, we find that in general $\langle (\Tr X)^P \rangle$ for $P>N$ can be expressed in terms of the $N$ unknowns $\langle (\Tr X)^p \rangle$ for $1\le p\le N$.

Although it is not yet clear, one can show that there are no further relations between the unknowns $\langle (\Tr X)^p \rangle$, $1\le p\le N$. Thus, non-perturbatively the loop equations do not completely solve the cubic matrix model. Instead, the correlators depend on $N$ additional parameters which were invisible in perturbation theory.

To determine the nature of these parameters, we consider the case $N=1$, corresponding to the integral:
\be \label{eqn:cubicN1}
Z = \frac{1}{2 \pi} \int_{-\infty}^{\infty} \dd x \, e^{-\frac{1}{2} x^2 - \lambda x^3} \,.
\ee
However, we notice an immediate problem: this integral diverges unless $\Re \lambda = 0$! This problem, arising from the factor of $e^{-(\Re \lambda) x^3}$ at either $x \to \infty$ or $x \to -\infty$, is invisible in perturbation theory, but makes the Hermitian matrix model ill-defined non-perturbatively (unless $\lambda$ is imaginary).

Recall that the loop equations (hence also perturbation theory) are insensitive to the choice of integration contour. A natural way to define a non-perturbative completion of the model is to choose a different integration contour $C$ such that the integral~(\ref{eqn:cubicN1}) is finite. To avoid introducing boundary terms into the loop equations, the integrand must vanish on $\partial C$. This occurs asymptotically at $|x| \to \infty$ with $|\arg(\lambda x^3)| < \frac{\pi}{2}$ (or with $|\arg(\lambda x^3)| = \frac{\pi}{2}$ and $|\arg(x^2)| < \frac{\pi}{2}$). There are three such regions, centered on $\arg(\lambda^{1/3} x) = 0, \pm \frac{2 \pi}{3}$, so there are two linearly-independent closed contours $C_1, C_2$ connecting these regions, see Figure~\ref{fig:cubiccontours}.
\begin{figure}
  \centering
  \includegraphics[width=0.25\textwidth]{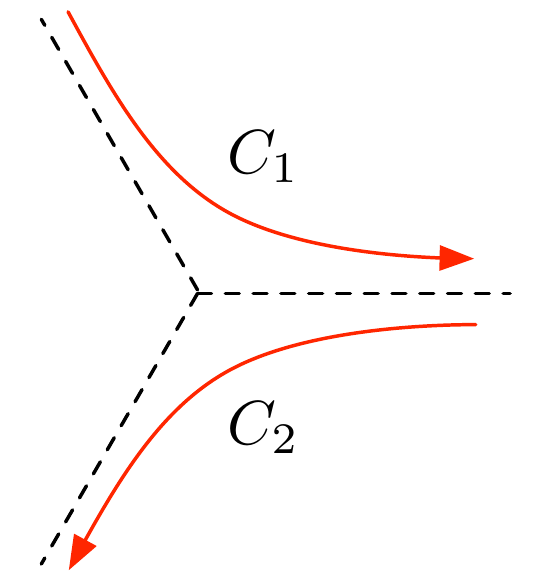}
  \caption{A basis of possible integration contours for a cubic potential.\label{fig:cubiccontours}}
\end{figure}
 For a general linear combination $C = c_1 C_1 + c_2 C_2$, one can check that the integrals $I_0 \equiv \int \dd x \, e^{-\frac{1}{2} x^2 - \lambda x^3}$ and $I_1 \equiv \int \dd x \, x e^{-\frac{1}{2} x^2 - \lambda x^3}$ depend on different linear combinations of $c_{1,2}$, hence a choice of $I_{0,1}$ is equivalent to a choice of contour. Since $Z = I_0 / 2 \pi$ and $\langle \Tr X \rangle = I_1 / I_0$, the information not specified by the loop equations is precisely the choice of integration contour.

Thus, we conclude that the choice of integration contour $C$ is a \emph{non-perturbative} ambiguity in the $N=1$ cubic matrix model, and this ambiguity precisely accounts for the missing information in the loop equations.  The need to introduce contours in field space to discuss the theory non-perturbatively

To extend this analysis to $N>1$, it is convenient to re-express the random matrix $X$ in terms of its $N$ eigenvalues $x_i$, $i=1,\ldots, N$. To do so, we apply the Faddeev-Popov gauge-fixing procedure to the $U(N)$ gauge symmetry of the Hermitian matrix model. The result is the eigenvalue model:
\be \label{eqn:Zeigenvalue}
Z = \frac{1}{N!} \int \prod_{i=1}^N \frac{\dd x_i}{2 \pi} \prod_{i<j} (x_i - x_j)^2 e^{- \sum_i W(x_i)} \,,
\ee
where $\prod_{i<j} (x_i - x_j)^2$ --- the square of the Vandermonde determinant $\det x_i^{j-1}$ --- is the gauge-fixing determinant, and we specialize to a single-trace potential $S[X] = \Tr W(X)$ for simplicity. The residual gauge-symmetry is $S_N \rtimes U(1)^N$, with volume $(2 \pi)^N N!$, where $U(1)^N$ acts trivially and $S_N$ permutes the eigenvalues. Gauge invariant operators are (sufficiently regular) symmetric functions of the $x_i$, the simplest class of which are symmetric polynomials, corresponding to the vector space generated by multitrace operators in the original matrix model.

The permissible integration contours for the cubic model $W(x) = \frac{1}{2} x^2 + \lambda x^3$ are the same as those for $N=1$ described above, except that the contour for each eigenvalue can be chosen separately. For $S_N$ invariant insertions, only the number of eigenvalues integrated along each contour will affect the answer, hence there is a basis $C_{N_1, N_2}$ of integration contours with $N_1$ eigenvalues integrated along $C_1$ and $N_2$ eigenvalues integrated along $C_2$. A general contour takes the form:
\be
C = \sum_{i=0}^N c_{i, N- i} C_{i, N-i} \,.
\ee
The insertions $I_p \equiv Z[(\Tr X)^p]$ for $0 \le p \le N$ will depend on linearly independent combinations of the $c_{i, N-i}$,\footnote{Assume the opposite. This implies the existence of a linear combination of contours such that $Z[(\Tr X)^p]=0$ for $0 \le p \le N$, hence by the loop equations $Z[\mathcal{O}]=0$ for any polynomial operator. In particular, $[1]=0$ and $Z[\Tr X^3] = 0$, the latter implying that $Z$ is independent of $\lambda$. This is a contradiction, because $Z[1] = 0$ at $\lambda = 0$ is incompatible with an analysis of the Gaussian.} so that the data not fixed by the loop equations exactly correspond to the choice of integration contour. As before, this is a non-perturbative ambiguity (invisible in perturbation theory).

Thus, the cubic eigenvalue model is sensitive to integration contour at the non-perturbative level, with $N+1$ independent possible contours. Except in the special case where $\lambda$ is purely imaginary, the real axis is not a possible choice of integration contour.

The quartic model $W(x) = \frac{1}{2} x^2 + \lambda x^4$ is an interesting counterpoint. In this case, the integral along the real axis converges for $\Re \lambda \ge 0$, hence there is a ``canonical'' choice of integration contour. Nonetheless, other integration contours are available --- a total of $\frac{(N+1)(N+2)}{2}$ are linearly independent --- with the same perturbation series and loop equations. In this sense, non-perturbative ambiguities persist, and the canonical resolution of these ambiguities is just one of many possibilities.

\subsection{Reflection positivity}

Nonetheless, the real axis is a distinguished contour, because only for this contour is the eigenvalue integral equivalent to an integral over \emph{Hermitian} matrices ($X^{\ast} = X$), with the measure~(\ref{eqn:Hermitianmeasure}). In principle, for other contours the eigenvalue integral can be written as a matrix integral of the holomorphic form
\be
\DD X = \bigwedge_{i,j=1}^N \dd X_{i j} \,,
\ee
over the cycle defined by the $U(N)$ orbit of the eigenvalue contours. This cycle can be defined in a $U(N)$-invariant way by equations of the form
\be
[X, X^\ast] = 0 \qquad, \qquad f_i(\Tr X, \ldots \Tr X^N) = 0 \,,
\ee
where the first equation specifies that $X$ is normal and the $f_i$ are $N$ real functions of $\Tr X, \ldots, \Tr X^N$ and their conjugates which specify the eigenvalue contour implicitly. For instance, the contour described by:\footnote{The constraints on $\Tr X^{2k}$, $k=1,\ldots,N$ can in principle be rewritten as constraints on $\Tr X, \ldots, \Tr X^N$.}
\be
[X, X^\ast] = 0 \qquad, \qquad \Im (\Tr X^{2k}) = 0\,, \qquad (k =1,\ldots, N)\,,
\ee
is relevant to the quartic model. Equivalently, this contour consists of normal matrices $X$ whose squares are Hermitian.

Despite the fact that $X \ne X^{\ast}$ on a general contour, the eigenvalue model --- and the related matrix integral --- retains some of the formal properties of a Hermitian matrix model. To illustrate this, we define an antilinear involution $\dag$ on the operator algbera by $X^{\dag} \equiv X$ (noting that in general $X^\dag \ne X^{\ast}$). The action of $\dag$ on an arbitrary operator is specified by antilinearity together with $(\Tr X^p)^\dag = (\Tr X^p)$.
 Provided that we choose a real potential and a contour satisfying $C = C^\ast$, this implies the formal property
\be
Z[\mathcal{O}^\dag] = Z[\mathcal{O}]^\ast \,,
\ee
for any operator $\mathcal{O}$.

To distinguish between eigenvalue models with these formal reality properties and an actual integral over Hermitian matrices, we note that the latter satisfies reflection positivity:
\be
Z[\mathcal{O}^\dag \mathcal{O}] > 0 \qquad \text{for any operator $\mathcal{O} \ne 0$.}
\ee
By contrast, the cubic model is in general \emph{not} reflection-positive.
Consider the $N=1$ model, for example. Reflection positivity requires that $\langle A_i^\dag A_j \rangle$ is a positive-definite matrix for any set of linearly-independent operators $\{ A_i \}$. Choosing the operators $\{1, \Tr X, \Tr X^2\}$ and applying the loop equations, we find that the matrix $\langle A_i^\dag A_j \rangle$ has at least one non-positive eigenvalue unless $|\lambda| < 2^{-1/2}\, 3^{-7/4} \simeq 0.1$. The constraint on $\lambda$ becomes successively tighter as we consider larger operators ($\Tr X^p$ for $p>2$) and indeed the non-polynomial operator:
\be
\mathcal{O} = (\Tr X) e^{\frac{1}{4} \Tr X^2}
\ee
satisfies $\langle \mathcal{O}^\dag\mathcal{O} \rangle = 0$, independent of $\lambda \ne 0$. Thus, the $N=1$ model is not reflection-positive, regardless of the choice of integration contour.\footnote{Technically, we could restrict the operator algebra to polynomial operators, eliminating this problematic operator, but for any fixed value of $\lambda \ne 0$, polynomial operators of finite degree will nonetheless violate reflection-positivity.} We expect the same to be true for $N>1$.

The cubic matrix model is therefore analogous to a non-unitary QFT. The quartic matrix model, by contrast, is manifestly reflection-positive when the integration contour is chosen to be the real line. A similar analysis to above shows that other contours are not reflection-positive, with increasingly large operators required to violate reflection-positivity as the integration contour approaches the real line. Thus, while the quartic integrated along the real line is analogous to a unitary QFT, other integration contours behave like non-unitary QFTs.

\subsection{The weak- and strong-coupling expansions}
\label{expans}

The existence of non-perturbative ambiguities is closely related to the divergence of perturbation theory, which typically defines only an asymptotic series near an essential singularity on the Riemann sphere. As an example, consider the $N=1$ quartic model, integrated along the real axis:
\be \label{eqn:quarticN1}
Z = \frac{1}{2\pi} \int_{-\infty}^{\infty} \dd x\, e^{-\frac{1}{2} x^2 - \lambda x^4} \,.
\ee
Expanding the partition function in powers of $\lambda$, we obtain formally:
\be \label{eqn:formalZseries}
Z = \frac{1}{2\pi} \sum_{p=0}^{\infty} \frac{(-1)^p}{p!} \lambda^p \int_{-\infty}^{\infty} \dd x\, x^{4 p} e^{-\frac{1}{2} x^2} = \frac{1}{\sqrt{2 \pi}} \sum_{p=0}^{\infty} (-1)^p \frac{(4p-1)!!}{p!} \lambda^p
\ee
where $n!! \equiv n (n-2) (n-4) \ldots$. Since
\be
\frac{(4p-1)!!}{p!} = \frac{2^{2p} \Gam{2p+\frac{1}{2}}}{\Gamma(p+1) \Gam{\frac12}} \sim p^p\,, \qquad (p \gg 1)\,,
\ee
we conclude that the radius of convergence of the formal perturbation series~(\ref{eqn:formalZseries}) is zero. Similar divergences appear with insertions and in normalized correlators.

Heuristically, perturbation theory diverges because for $|x| \gg 1/\sqrt{|\lambda|}$ the quartic coupling dominates the integral~(\ref{eqn:quarticN1}), hence a perturbative expansion in $\lambda$ is not justified. In particular, the integral diverges for $\Re \lambda < 0$, whereas it converges for $\Re \lambda > 0$, implying that $\lambda = 0$ is an essential singularity in the holomorphic function $Z(\lambda)$. This is similar to how contour dependence appears non-perturbatively. Since the quartic coupling dominates for $|x| \gg 1/\sqrt{|\lambda|}$, there are additional integration contours where the quadratic term $e^{-\frac{1}{2} x^2}$ diverges but the quartic term keeps the integral finite --- such as the imaginary axis --- and the choice of integration contour generates a non-perturbative ambiguity.\footnote{Resolving these non-perturbative ambiguities in unitary theories is an active research topic, see, e.g.,~\cite{Dunne:2015eaa}.}

Thus, these twin problems of perturbation theory --- divergence and insensitivity to non-perturbative physics --- are both linked to the dominance of interactions at large field values. A novel approach would be to instead expand the exponential in the quadratic coupling, keeping the interaction term fixed. To do so, we rescale $x \to x / \lambda^{1/4}$ to obtain the model
\be \label{eqn:quarticN1rescaled}
Z = \frac{1}{2\pi} \int_{-\infty}^{\infty} \dd x\, e^{-\frac{1}{2 g^2} x^2 - x^4} \,,
\ee
up to a normalizing factor for $Z$, where $g = \sqrt{\lambda}$. If we now expand the exponential about the \emph{strong-coupling limit}, $g \to \infty$, we obtain
\be
Z = \frac{1}{2\pi} \sum_{p=0}^\infty \frac{(-1)^p}{2^p p!} g^{-2 p} \int_{-\infty}^{\infty} \dd x\, x^{2 p} e^{- x^4} = \frac{1}{4 \pi} \sum_{p=0}^\infty \frac{(-1)^p \Gam{\frac{p}{2}+\frac{1}{4}}}{2^p p!} g^{-2 p} \,.
\ee
Using Stirling's approximation, we conclude that
\be
\frac{\Gam{\frac{p}{2}+\frac{1}{4}}}{2^p p!} \sim p^{-p/2} \,,
\ee
hence the perturbation series converges! The sum can be performed explicitly,
\be
Z = \frac{1}{4 \pi} \sum_{p=0}^\infty \frac{(-1)^p \Gam{\frac{p}{2}+\frac{1}{4}}}{2^p p!} g^{-2 p} = \frac{1}{\sqrt{32 \pi^2 g^2}}\, e^{\frac{1}{32 g^4}} K_{\frac{1}{4}}\!\left(\frac{1}{32 g^4}\right) \,,
\ee
where $K_{\nu}(z)$ is the modified Bessel function of the second kind, a result which can be verified by direct integration of~(\ref{eqn:quarticN1rescaled}).

Provided that a solution in the strong coupling limit, $g \to \infty$, is available, the strong-coupling expansion $g \gg 1$ has much better properties than the weak-coupling expansion $g \ll 1$. In particular, the value of $g \ne 0$ does not affect the convergence of the partition function, so we expect that $Z(1/g^2)$ is analytic at $1/g^2 = 0$, and the expansion in $g\gg 1$ should converge. For the same reason, there are no analogs of non-perturbative ambiguities. Indeed, the pure quartic model, $g = \infty$, depends on the same set of contours as at any intermediate coupling $g\ne 0$, so the ambiguities that went unresolved perturbatively at weak coupling are fixed at strong coupling, even before perturbing!

In the next section, we explore the feasibility of solving the cubic or quartic model in the strong coupling limit $g \to \infty$ for arbitrary finite $N$, which would enable a solution for any $g$ via the strong-coupling expansion described above.

\section{Exact Solutions at Strong Coupling} \label{sec:exactsolutions}

Motivated by the above considerations, we analyze one-matrix models in the limit where an $r$-point coupling blows up, $\lambda_r \to \infty$. After a field redefinition, these correspond to monomial matrix models $S[X] = \Tr X^r$. We begin by considering the quartic with a real (reflection-positive) integration contour, before generalizing to other contours and potentials.

\subsection{The real-line quartic} \label{sec:quartic}

We consider the pure quartic model:
\be
Z_N  = \frac{1}{N!} \int_{-\infty}^{\infty} \prod_{i=1}^N \frac{\dd x_i}{2 \pi} \, e^{-x_i^4} \prod_{i>j} (x_i - x_j)^2 \,.
\ee
We can rewrite the partition function as a determinant using a standard trick; note that
\be
\prod_{i>j} (x_i - x_j) = \det V_{i j} = \det x_i^{j-1} \,,
\ee
where $V_{i j} = x_i^{j-1}$ is the Vandermonde matrix. Fixing the $S_N$ permutation symmetry of the partition function, we obtain
\be
Z_N  = \int_{-\infty}^{\infty} \prod_{i=1}^N \frac{\dd x_i}{2 \pi} \, x_i^{i-1} e^{-x_i^4} \times (\det x_j^{k-1}) = \det_{N \times N} Z_1[x^{i+j-2}] \,,
\ee
where $\det_{N \times N}$ denotes the determinant of the upper-left $N\times N$ block, $i,j=1,\ldots,N$. Evaluating the integral directly, we find
\be \label{eqn:Rquartic1}
Z_1[x^p] = \int_{-\infty}^{\infty} \frac{\dd x}{2 \pi}\, x^p e^{-x^4} = \begin{cases} \frac{1}{4 \pi} \Gam{\frac{p+1}{4}} & p \in 2 \bZ \\ 0 & p \in 2 \bZ+1\end{cases} \,.
\ee
Thus,
\be \label{eqn:quarticdet}
Z_N = \frac{1}{(4 \pi)^N} \det_{N \times N} \frac{1+(-1)^{i+j}}{2}\cdot\Gam{\frac{i+j-1}{4}} \,.
\ee
By similar reasoning
\be
Z_N\!\left[\prod_{a=1}^n \Tr X^{p_a}\right] = \frac{1}{(4 \pi)^N} \sum_{k_1, \ldots, k_n = 1}^N \det_{N \times N} \frac{1+(-1)^{i+\sum_a p_a \delta_{i k_a}+j}}{2}\cdot\Gam{\frac{i+\sum_a p_a \delta_{i k_a}+j-1}{4}} \,.
\ee

These relatively simple results mask the complexity of the model inside a determinant. 
To illustrate this complexity, we evaluate the partition function for small values of $N$:
\begin{multline} \label{eqn:RquarticZn}
Z_{0,1,\ldots } = \Bigg\{1, \frac{\GamF{1}{4}}{4 \pi}, \frac{\GamF{1}{4} \GamF{3}{4}}{16 \pi^2}, \frac{\GamF{1}{4}^2 \GamF{3}{4} - 4\, \GamF{3}{4}^3}{256 \pi^3}, \frac{-\GamF{1}{4}^4 + 16\, \GamF{1}{4}^2 \GamF{3}{4}^2 - 48\, \GamF{3}{4}^4}{2^{14}\, \pi^4}, \frac{-\GamF{1}{4}^5+20\, \GamF{1}{4}^3 \GamF{3}{4}^2-96\, \GamF{1}{4} \GamF{3}{4}^4}{2^{18}\, \pi^5}, \\
\frac{-\GamF{1}{4}^5 \GamF{3}{4}+17\, \GamF{1}{4}^3 \GamF{3}{4}^3-72\, \GamF{1}{4} \GamF{3}{4}^5}{2^{20}\, \pi^6}, \frac{-5\, \GamF{1}{4}^6 \GamF{3}{4}+105\, \GamF{1}{4}^4 \GamF{3}{4}^3-684\, \GamF{1}{4}^2 \GamF{3}{4}^5+1296\, \GamF{3}{4}^7}{2^{26}\, \pi^7}, \ldots 
\Bigg\}
\end{multline}
where we use the shorthand $\Gamma_x \equiv \Gamma(x)$. Here we have simplified the determinant by reducing $\Gam{\frac{2k+1}{4}}$ to an rational number times either 
$\Gam{\frac{1}{4}}$ or $\Gam{\frac{3}{4}}$; since $\Gam{\frac{1}{4}}/\Gam{\frac{3}{4}} \simeq 2.958675$ is trancendental, there are no obvious further simplifications. The same complexity is evident in the normalized correlators, for instance
\be
\langle \Tr X^2 \rangle_{0,1,\ldots} = \bigg\{ 0 , \frac{\GamF{3}{4}}{\GamF{1}{4}}, \frac{\GamF{1}{4}^2 + 4\, \GamF{3}{4}^2}{4\, \GamF{1}{4} \GamF{3}{4}}, \frac{\GamF{1}{4}^3 + 4\, \GamF{1}{4}\GamF{3}{4}^2}{4\, \GamF{1}{4}^2 \GamF{3}{4}-16 \, \GamF{3}{4}^3},  \frac{16\, \GamF{1}{4}\GamF{3}{4}^3}{-\GamF{1}{4}^4+16\, \GamF{1}{4} \GamF{3}{4}-48 \, \GamF{3}{4}^4}, \ldots \bigg\} \,,
\ee
with increasingly complicated expressions for larger $N$.

Nonetheless, the result~(\ref{eqn:RquarticZn}) has some obvious structure. The partition function takes the general form:
\be \label{eqn:RquarticZsectors}
Z_N = \frac{1}{(2\pi)^N} \sum_{i=0}^N a_{i;N} \GamF{1}{4}^i \GamF{3}{4}^{N-i}\,,
\ee
where the coefficients $a_{i;N}$ are rational. 
More generally, for any polynomial operator $\mathcal{O}$
\be \label{eqn:Rquarticsectors}
Z_N[\mathcal{O}] = \frac{1}{(2\pi)^N} \sum_{i=0}^N a_{i;N}(\mathcal{O})\, \GamF{1}{4}^i \GamF{3}{4}^{N-i} \,,
\ee
for rational coefficients $a_{i;N}(\mathcal{O})$, since the eigenvalue integral can be evaluated by expanding the Vandermonde determinant times $\mathcal{O}$ into a sum of monomials and applying~(\ref{eqn:Rquartic1}).

It is natural to interpret~(\ref{eqn:Rquarticsectors}) as a sum over sectors, $Z[\mathcal{O}] = \sum_i Z_i[\mathcal{O}]$, in which case the complexity of the real-line quartic model can be partially ascribed to the increasing number of sectors --- $N+1$ for $Z_N$. One approach to solving the model is to solve each sector individually. After identifying the origin of these sectors, we will show that at least some of them admit exact solutions for all $N$.

\subsection{Pure and mixed phases} \label{sec:puremixed}

We now generalize to the potential:
\be
S[X] = \Tr X^r \,,
\ee
which includes the Gaussian $(r=2)$, cubic $(r=3)$, and quartic $(r=4)$ as special cases. The corresponding $N=1$ matrix model admits $r-1$ closed contours on which the integral converges, constructed as follows. Let $B_{r,j}$ denote the contour from $0$ to $\omega_r^j\cdot\infty$, where $\omega_r \equiv e^{2\pi i/r}$. These contours are open, but represent all possible asymptotics for which the integral converges. We form the ``Fourier-transformed'' contours:
\be \label{eq:contours}
C_{r,a} \equiv \sum_{j = 0}^{r - 1} \omega_r^{- j a} B_{r,j} \,.
\ee
We observe that $\partial C_{r,a} = 0$ for $a \not\equiv 0 \mod r$, hence $C_{r,1},\ldots, C_{r,r-1}$ form a basis of closed contours on which the integral converges. This basis, illustrated in Figure~\ref{fig:factorizedcontours}, is the eigenbasis of the $\bZ_r$ symmetry $X \to \omega_r X$, which maps $C_{r,a} \to \omega_r^{-a} C_{r,a}$
\begin{figure}
  \centering
  \includegraphics[width=0.95\textwidth]{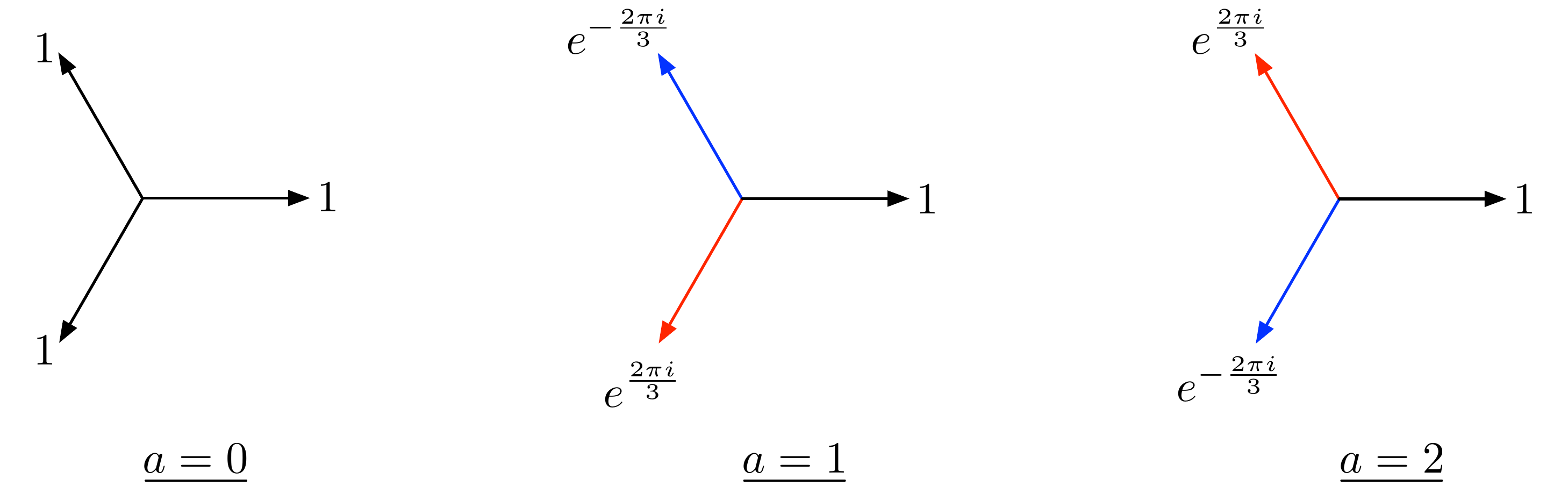}
  \caption{A natural basis of integration contours for a cubic ($r=3$) potential, where the values at the end of each ray (and the color of the ray) denote its weight within the contour. We include the case $a=0$, which is not a closed contour but occurs naturally once we include the $u$-deformation in \S\ref{sec:genconstruction}.\label{fig:factorizedcontours}}
\end{figure}

For $N>1$, a general contour can be written as a linear combination of
\be
C^{(r)}_{N_1, \ldots, N_{r-1}} \equiv C_{r,1}^{N_1}\times \ldots \times C_{r,r-1}^{N_{r-1}} \,, \qquad \, N = \sum_{a=1}^{r-1} N_a\,,
\ee
i.e.\ with $N_a$ eigenvalues integrated along the contour $C_{r,a}$. There are $\binom{N+r-2}{r-2}$ such contours. We refer to matrix models integrated over the contours $C_{r,a}^N$ as ``pure phases,'' and those integrated over $C_{N_1, \ldots, N_{r-1}}^{(r)}$ with $N_1,\ldots,N_r < N$ as ``mixed phases.''

The principle advantage of this contour basis is that it simplifies the moments:
\be \label{eqn:Cmoments}
\int_{C_a} x^p e^{-x^r} \dd x = \delta_{r|p+1-a}\, \Gam{\frac{p+1}{r}} \,,
\ee
where $\delta_{r|p} = 1$ when $r$ divides $p$, and vanishes otherwise. The non-vanishing moments are those for which the integrand together with the contour forms a $\bZ_r$ singlet, and each such moment is a rational prefactor times $\Gam{\frac{a}{r}}$.
As a consequence
\be
\int_{C_{N_1, \ldots, N_{r-1}}} \mathcal{O} \prod_i e^{-x_i^r} \sim \prod_a \GamF{a}{r}^{N_a} \,,
\ee
up to a rational prefactor, where $\mathcal{O}$ is any polynomial insertion. Comparing with~(\ref{eqn:Rquarticsectors}), we see that the sum over sectors previously identified in the real-line quartic is nothing but a sum over pure and mixed phases!

In particular, for $r=4$, the contour along the real axis is $R = \frac{1}{2} (C_{4,1} + C_{4,3})$, where
\be
R^N = \frac{1}{2^N} \sum_{i=0}^N \binom{N}{i} C_{4,1}^i C_{4,3}^{N-i} \,.
\ee
Not every phase mixture contributes to the partition function, as the $\bZ_r$ symmetry dictates that many insertions vanish. Suppose that $\mathcal{O}_p$ is a homogeneous polynomial in $X$ of degree $p$. A necessary condition for $Z_{N_1,\ldots,N_{r-1}}[\mathcal{O}_p]$ to be non-vanishing is for it to be a $\bZ_r$ singlet:
\be
p + N^2 - \sum_a a N_a \equiv 0 \mod r \,.
\ee
For instance, the partition function $(p=0)$ of a pure phases vanishes unless $N = k r$ or $N = k r+a$. For the $C_{4,1}, C_{4,3}$ two-phase mixture of the real-line quartic, we obtain the constraint
\be
N_1 \equiv \frac{N (N+1)}{2} \mod 2 \,,
\ee
for contributions to the partition function, which explains the non-vanishing terms in~(\ref{eqn:RquarticZn}).

Having identified the natural subsectors of the real-line quartic, the obvious question is whether these subsectors are solvable. Remarkably, as we now argue, the pure phases are exactly solvable for any $r$, $a$ and $N$.

\subsection{Summary of results for pure phases}

Before explaining how the pure phases can be solved, we present the solution in brief. We consider the pure-phase eigenvalue model
  \be \label{eqn:modeldef}
  Z_N^{(r,a)} = \frac{1}{N!} \int_{C_{r,a}} \prod_i \frac{\dd x_i}{2 \pi}\, e^{-x_i^r} \prod_{i>j} (x_i - x_j)^2 \,,
  \ee
  for integers $N \ge 0$, $r>0$, and $0 \le a <r$, with $C_{r,a}$ defined by (\ref{eq:contours}). As usual, insertions and expectation values are defined as in~(\ref{eqn:Zintegral}--\ref{eqn:ZOintegral-normalized}).

\begin{theorem} \label{thm:ZNra}
The partition function of the monomial matrix model~(\ref{eqn:modeldef}) is
\be \label{eqn:ZNra}
\boxed{Z_N^{(r,a)} = \frac{\delta_{r,a}(N)}{(2\pi)^N} \prod_{i=0}^{N-1}  \Gam{\floor{\frac{i}{r}}+1} \Gam{\floor{\frac{i-a}{r}}+\frac{a}{r}+1} \,,}
\ee
where $\delta_{r,a}(N) = 0,\pm 1$ is given by
\be \label{eqn:deltaraN}
\delta_{r,a}(N) = \begin{cases} (- 1)^{\ceil{\frac{N}{r}}  \frac{a (a - 1)}{2} +
  \left\lfloor \frac{N}{r} \right\rfloor  \frac{\tilde{a} (\tilde{a} - 1)}{2}} & N \equiv 0,a \mod r\\ 0 & \text{otherwise} \end{cases} \,,
\ee
and $\tilde{a} \equiv r-a$.
\end{theorem}

Recall that Schur polynomials are multivariate symmetric polynomials defined as
\be \label{eqn:Schurdef}
s_\lambda(x) \equiv \frac{\det x_i^{N+\lambda_j-j}}{\det x_i^{N-j}} \,,
\ee
where $\lambda$ denotes a partition $\lambda_1 \ge \ldots \ge \lambda_N \ge 0$.

\begin{theorem} \label{thm:1schur}
The averages of Schur polynomials take the simple form:
\begin{equation} \label{eqn:1schur}
\boxed{ \langle s_{\lambda} (X) \rangle = \frac{\delta_r(\lambda)}{r^{|\lambda|/r}} \prod_{x \in \lambda} \frac{\llbracket N + c_\lambda(x) \rrbracket_{r,0}\, \llbracket N + c_\lambda(x) \rrbracket_{r,a}}{\llbracket h_\lambda(x) \rrbracket_{r,0}} \,,}
\end{equation}
where
\begin{equation} \label{eqn:bracketdef1}
\llbracket n \rrbracket_{r, a} = \begin{cases} n & n \equiv a \mod r \\ 1 & \text{otherwise} \end{cases} \,.
\end{equation}
\end{theorem}
\noindent Here we interpret $\lambda$ as a Young diagram with $|\lambda|$ boxes $x=(i,j)$, rows of length $\lambda_i$ and columns of length $\lambda_i^{\T}$. The contents $c_\lambda(x) \equiv j - i$ and hook-length $h_\lambda(x) \equiv \lambda_i + \lambda_j^{\T} - i - j +1$ are the same quantities which appear in the dimension formula for representations of the special linear group:
\begin{equation}
\dim_{SL(N)}(R_\lambda) = \prod_{x\in\lambda} \frac{N+c_\lambda(x)}{h_\lambda(x)} \,,
\end{equation}
which is similar to~(\ref{eqn:1schur}). Finally, the prefactor $\delta_r(\lambda) = 0,\pm1$, the ``$r$-signature'' of $\lambda$, is given explicitly by
\begin{equation} \label{eqn:deltarlambda}
\delta_r(\lambda) = \begin{cases} (-1)^{\frac{|\lambda|}{r}}\prod_{x\in\lambda} (-1)^{\floor{\frac{c_\lambda(x)}{r}}+\floor{\frac{h_\lambda(x)}{r}}} & \text{The $r$-core of $\lambda$ is trivial} \\ 0 & \text{otherwise} \end{cases} \,.
\end{equation}
Here the $r$-core of $\lambda$---the unique result of stripping all possible rim hooks of length $r$ from $\lambda$---generalizes the remainder upon division by $r$, so that Schur polynomials with non-zero averages correspond to Young diagrams which are ``divisible by $r$,'' or ``$r$-divisible.'' A necessary but insufficient condition for $r$-divisibility is that $|\lambda| \equiv 0 \mod r$. The relevant properties of $r$-cores, $r$-signatures, and $r$-divisible Young diagrams are reviewed in Appendix~\ref{app:schurdivisibility}.

We note in passing that the Schur average formula~(\ref{eqn:1schur}) obeys several non-trivial relations when treated as a formal analytic function of $N$ for fixed $N \bmod r$:
\begin{align} \label{eqn:schurRelns}
  \langle s_{\lambda} (X) \rangle_{N}^{(r, a)} &= \langle s_{\lambda} (X) \rangle_{N'}^{(r, r - a)} \,, &
  \langle s_{\lambda} (X) \rangle_{- N}^{(r, a)} &= (- 1)^{(r+1)\frac{|\lambda |}{r}}  \langle s_{\lambda^{\T}} (X) \rangle_N^{(r, r - a)} \,,
\end{align}
where $N'$ indicates the opposite value of $N \bmod r$. The first equation relates the two possible values of $N \bmod r$ for fixed $N$, whereas the second equation relates $N \to -N$ for fixed $N \bmod r$, similar to negative rank duality (see, e.g.,~\cite[Ch.\ 13]{Cvitanovic:2008zz}).

Since any symmetric polynomial in $x_i$, $i=1,\ldots,N$, can be expressed in terms of the Schur polynomials, in principle~(\ref{eqn:ZNra}) and~(\ref{eqn:1schur}) provide a complete solution to the matrix model~(\ref{eqn:modeldef}) for any $N$.\footnote{Our formulas for normalized correlators apply when $Z \ne 0$. Unnormalized insertions can still be finite when $Z = 0$, but besides a brief discussion in~\S\ref{subsec:refactorization}, we leave a thorough treatment of these cases to the future.} For instance, single-trace correlation functions are calculated explicitly in~\S\ref{sec:singletrace}, see (\ref{eqn:singletraceA}), (\ref{eqn:singletraceB}), (\ref{eqn:singletraceEx}) and (\ref{eqn:singletraceN}).

\subsection{Solution by orthogonal polynomials} \label{subsec:orthopolys}

We now prove Theorems \ref{thm:ZNra} and \ref{thm:1schur},
i.e.\ derive formulas \eqref{eqn:ZNra} and \eqref{eqn:1schur},
using orthogonal polynomials. Suppose that $p_n (x) = x^n + \ldots$ is some polynomial basis. The
Vandermonde determinant can be rewritten as $\det x_i^{j - 1} = \det p_{j - 1} (x_i)$ by a triangular change of basis,
so that
\begin{equation}
  Z_N = \det_{N\times N} Z_1[p_i p_j] \,.
\end{equation}
If we choose a polynomial basis for which the matrix $Z_1[p_i p_j]$ is sufficient sparse then the partition function is easily computed. The usual
approach is to choose $Z_1[p_i p_j] = t_i \delta_{i j}$ for some
normalization $t_i$, so that $Z_N = \prod_{i = 0}^{N - 1} t_i$.
This approach is well-suited to the Gaussian model ($r = 2$ and $a=1$), but is impossible in the pure phases for $r > 2$, because $Z_{N = k
r}$ and $Z_{N = k r + a}$ do not vanish, whereas $Z_N$ vanishes for other
values of $N$, implying that some of the $t_i$ vanish and others are infinite.

Instead, we consider orthogonal polynomials satisfying the orthogonality relation:
\begin{equation} \label{eqn:orthocond}
  Z_1[ p_m p_n ] = t_m \delta^{(r, a)}_{m, n} \,,
\end{equation}
where $\delta^{(r, a)}_{m, n}$ is the block-diagonal matrix $\diag(J_a, J_{r-a}, J_a, J_{r-a},\ldots)$ with $J_n$ the $n\times n$ antidiagonal permutation matrix. This condition is chosen so that $\delta^{(r,a)}_{m,n}$ is nonzero on a subset of the nonzero entries of $Z_1[x^m x^n]$, i.e., those satisfying $m+n \equiv a-1 \pmod r$.
The solution is
\begin{equation} \label{eqn:orthoSol}
p_{r k+i}(x) = \begin{cases}  x^i \hat{L}_k^{\left( \frac{a}{r} - 1 \right)}(x^r) & 0 \le i < a \\
  x^i \hat{L}_k^{\left( \frac{a}{r} \right)} ( x^r ) & a \le i < r\end{cases} \,,
\end{equation}
where $L_k^{(\alpha)}(x)$ denotes the generalized Laguerre polynomial
\begin{equation}
  L_k^{(\alpha)} (x) = \sum_{p = 0}^k (-1)^p  \begin{pmatrix} k + \alpha\\ k - p \end{pmatrix} \frac{x^p}{p!}\,,
\end{equation}
and $\hat{L}_k^{(\alpha)}(x) = (-1)^k k! L_k^{(\alpha)}(x) = x^k + \ldots$ is monic. The Laguerre polynomials satisfy the orthogonality relation:
\begin{equation} \label{eqn:orthoLag}
\int_0^\infty x^\alpha L_m^{(\alpha)}(x) L_n^{(\alpha)}(x) e^{-x} \dd x = \frac{\Gamma(n+\alpha+1)}{n!} \delta_{m n} \,.
\end{equation}
Together,~(\ref{eqn:orthoSol}) and~(\ref{eqn:orthoLag}) are sufficient to derive~(\ref{eqn:orthocond}), where the homogeneity property $p_n (\omega_r x) = \omega_r^n p_n (x)$ implies that the integral over $C_{r,a}$ either vanishes or is equivalent to an integral over $B_0=(0,\infty)$, which can be reduced to~(\ref{eqn:orthoLag}) by a change of variables $y = x^r$. We obtain
\begin{equation} \label{eqn:ti}
t_i = \frac{1}{2 \pi} \Gam{\floor{\frac{i}{r}}+1} \Gam{\floor{\frac{i-a}{r}}+\frac{a}{r}+1} \,.
\end{equation}
The partition function is therefore $Z_N = \delta_{r,a}(N) \prod_{i=0}^{N-1} t_i$, where $\delta_{r,a}(N) = \det_{N \times N} \delta^{(r, a)}_{m, n}$ is the determinant of upper-left $N\times N$ block of $\delta^{(r, a)}_{m, n}$. This matches~(\ref{eqn:ZNra}).

The Schur polynomial average~(\ref{eqn:1schur}) can also be derived using orthogonal polynomials, as we now show.\footnote{We follow a similar approach to~\cite{Dolivet:2006ii}.} Our starting point is the formula~\cite[p.\ 67]{Macdonald:2008zz}
\begin{equation} \label{eqn:detexpansion}
  \prod_{j = 1}^k \det (z_j - X) = \sum_{\underset{\lambda_1 \le k}{\lambda}} (- 1)^{| \lambda |}
  s_{\lambda} (X) s_{\tilde{\lambda}}(z) \,,
\end{equation}
where $\tilde{\lambda}$ is the partition $\tilde{\lambda_i} = (N - \lambda^{\T}_k, \ldots, N - \lambda^{\T}_1)$. It is straightforward to check using the orthogonality relation~(\ref{eqn:orthocond}) that:
\begin{equation} \label{eqn:detformula}
  \left\langle \prod_{j = 1}^k \det (z_j - X) \right\rangle = \frac{1}{\underset{1\le i,j \le k}{\det} z_i^{j-1}}\, {\det}
  {\begin{pmatrix}
    p_N (z_1) & \ldots & p_{N + k - 1} (z_1)\\
    \vdots &  & \vdots\\
    p_N (z_k) & \ldots & p_{N + k - 1} (z_k)
  \end{pmatrix}} \,.
\end{equation}
In general, this holds when $Z_1[p_i p_j]= 0$ for $i< N$ and $j \ge N$, which follows from~(\ref{eqn:orthocond}) when $Z_N\ne 0$, i.e., when $N=k r$ or $N=kr+a$.

Combing~(\ref{eqn:detexpansion}) with~(\ref{eqn:detformula}) and using $s_{\tilde{\lambda}}(z) = \det z_i^{N-\lambda^\T_j+j-1}/\det z_i^{j-1}$, we obtain
\begin{equation} \label{eqn:detp1}
\underset{1\le i,j \le k}{\det} p_{N+j-1}(z_i) = \sum_{\underset{\lambda_1 \le k}{\lambda}} (-1)^{|\lambda|} \langle s_\lambda(X) \rangle \underset{1\le i,j \le k}{\det}  z_i^{N- \lambda^\T_j+j-1} \,.
\end{equation}
For orthogonal polynomials of the general form $p_i = \sum_j p_{i; j} x^j$, we have
\begin{equation}
  \underset{1 \le i, j \le k}{\det} p_{N + i - 1} (z_j) =
  \sum_{p_1, \ldots, p_k} \left( \prod_{j = 1}^k z_j^{p_j} \right)
  \underset{1 \le i, j \le k}{\det} p_{N + i - 1;\, p_j} \,,
\end{equation}
so that
\begin{equation} \label{eqn:polyschurformula}
  \langle s_{\lambda} (X) \rangle = (- 1)^{| \lambda |} \underset{1 \leqslant
  i, j \leqslant k}{\det} p_{N + i - 1;\, N - \lambda_j^{\T}+ j - 1 } \,, \qquad k \ge \lambda_1\,.
\end{equation}

The general result~(\ref{eqn:polyschurformula}) can be applied to the case at hand by noting that
\be
p_{i;j} = \delta_{r|(i-j)} \frac{(-1)^{\frac{i-j}{r}}}{\bigl(\frac{i-j}{r}\bigr)!}\, \frac{t_i}{t_j} \,,
\ee
where the Kronecker delta enforces $\frac{i-j}{r} \in \mathbb{Z}$. We obtain
\begin{align}
 \langle s_\lambda(X) \rangle = (-1)^{|\lambda|} \left(\prod_{j = 1}^k \frac{t_{N + j -1}}{t_{N + j - 1 - \lambda^\T_j}} \right)
 \det_{1 \le i, j \le k} \delta_{r|(\lambda^\T_j + i - j)} \frac{(-1)^{\frac{\lambda^\T_j + i - j}{r}}}{(\frac{\lambda^\T_j + i - j}{r})!} \,,
\end{align}
for $k\ge \lambda_1$.

To evaluate the determinant, we use the results of \S\ref{sec:determinants}. We must have
\be
\underset{1\le i,j \le k}{\det} \frac{(-1)^{\lambda^\T_j + i - j}}{(\lambda^\T_j + i - j)!} = (-1)^{|\lambda|} \prod_{x \in \lambda} \frac{1}{h_\lambda(x)} \,,  \qquad k \ge \lambda_1 \,,
\ee
to reproduce the Kadell formula for $r=1$~\cite{Kadell:1988}. The hook-lengths in the various components of the $r$-quotient $\lambda/r^{(\mu)}$ are just $1/r$ times the hook lengths divisible by $r$ in $\lambda$, so we obtain
\be
\det_{1 \le i, j \le k} \delta_{r|(\lambda^\T_j + i - j)} \frac{(-1)^{\frac{\lambda^\T_j + i - j}{r}}}{(\frac{\lambda^\T_j + i - j}{r})!} = r^{|\lambda|/r} \delta_r(\lambda) \prod_{x \in \lambda} \frac{1}{\llbracket h_\lambda(x) \rrbracket_{r,0}} \,,
\ee
using~(\ref{eqn:exDetFormula}) and Theorem~\ref{thm:transpose}. Using (\ref{eqn:ti}) and the identity
\begin{equation} \label{eqn:modfactorial}
\frac{\Gam{\floor{\frac{n+I-a}{r}}+\frac{a}{r}+1}}{\Gam{\floor{\frac{n-a}{r}}+\frac{a}{r}+1}} = r^{-\left(\floor{\frac{n+I-a}{r}}-\floor{\frac{n-a}{r}}\right)} \prod_{i=1}^I \llbracket n + i \rrbracket_{r,a} \,,
\end{equation}
to simplify the product over $t$s,
 we obtain the Schur average formula~(\ref{eqn:1schur}).

\subsection{Single-trace correlators} \label{sec:singletrace}

Let $L^{(I,J)}$ denote the L-shaped partition with $I+1$ rows and $J+1$ columns, i.e.\ $L^{(I,J)}_1 = J+1$, $L^{(I,J)}_2 = \ldots = L^{(I,J)}_{I+1}=1$.
We have~\cite[p.\ 48]{Macdonald:2008zz}
\begin{equation}
\Tr X^p = \sum_{I=0}^{p-1} (-1)^I s_{L^{(I,p-1-I)}}(X) \,.
\end{equation}
Combining this with the Schur average formula~(\ref{eqn:1schur}), we compute the expectation values of single-trace operators.
Using
$\delta_r(L^{(I,J)})=\delta_{r|(I+J+1)} (-1)^{I+\floor{\frac{I}{r}}}$,
we obtain:
\begin{align}
\langle \Tr X^{q r} \rangle &=\frac{1}{q r^{q+1}} \sum_{J=0}^{qr-1} \frac{(-1)^{\floor{\frac{I}{r}}}}{\dbracket{I}_{r,0}! \dbracket{J}_{r,0}!} \prod_{i=-I}^{J} \dbracket{N+i}_{r,0} \dbracket{N+i}_{r,a} \,,
\end{align}
for $q>0$, where $I=q r-1-J$ and $\dbracket{n}_{r,a}! \equiv \prod_{i=1}^n \dbracket{i}_{r,a}$. By~(\ref{eqn:modfactorial}), this can be rewritten as
\begin{align}
\langle \Tr X^{q r} \rangle &=\frac{1}{q} \sum_{J=0}^{qr-1} \frac{(-1)^{\floor{\frac{I}{r}}} \Gam{\floor{\frac{N+J}{r}}+1}\Gam{\floor{\frac{N+J-a}{r}}+\frac{a}{r}+1}}{\Gam{\floor{\frac{N+J}{r}}-q+1}\Gam{\floor{\frac{N+J-a}{r}}+\frac{a}{r}-q+1} \floor{\frac{I}{r}}! \floor{\frac{J}{r}}!}\,.
\end{align}
Collecting terms, this takes the form of a sum of hypergeometric series
\begin{equation} \label{eqn:singletraceA}
\begin{aligned}
\langle \Tr X^{q r} \rangle_{N=kr} &= a f_{k,\frac{a}{r}}(q) +(r-a) f_{k,\frac{a}{r}+1}(q) \,,\\
\langle \Tr X^{q r} \rangle_{N=kr+a} &= a f_{k+1,\frac{a}{r}}(q)+ (r-a) f_{k,\frac{a}{r}+1}(q) \,,
\end{aligned}
\end{equation}
for $0<a<r$, where
\begin{equation}
f_{k,x}(q) \equiv \frac{(-1)^{q-1}}{q!} \sum_{j=0}^{q-1} \frac{\Gam{k+j+1} \Gam{k+j+x} (1-q)_j}{ \Gam{k+j-q+1} \Gam{k+j-q+x} j!} \,,
\end{equation}
and $(x)_n \equiv x (x+1) \ldots (x+n-1)$ denotes the rising factorial.
Using a pair of resummation identities for $\pfq{3}{2}$, this can be rewritten as\footnote{Note that the generating function $\sum_{q=0}^\infty \frac{\langle \Tr X^{q r} \rangle}{(x)_q}\, t^q$ can be written in terms of $\pfq{2}{1}$. However, the resulting expression is no easier to work with than (\ref{eqn:singletraceA}), (\ref{eqn:singletraceB}), and this curiosity will play no further role in our discussion.}
\begin{equation} \label{eqn:singletraceB}
f_{k,x}(q) =
 k\, (x)_q\; \pFq{3}{2}{1-k, 1+q, -q}{ x, 2}{1}\,.
\end{equation}

Combining (\ref{eqn:singletraceA}) and (\ref{eqn:singletraceB}), it is straightforward to compute any single-trace correlation function of interest. For instance,
\begin{equation} \label{eqn:singletraceEx}
\begin{aligned}
 r \langle \Tr X^r \rangle &= N^2 \,,\\
 r^2 \langle \Tr X^{2 r} \rangle &= 2 N^3 + a \tilde{a} N\, \\
 r^3 \langle \Tr X^{3 r} \rangle &= 5 N^4 + (r^2 + 6 a \tilde{a}) N^2 \pm a \tilde{a} (a - \tilde{a}) N\,, \\
 r^4 \langle \Tr X^{4 r} \rangle &= 14 N^5 + 10 (r^2 + 3 a \tilde{a}) N^3 \pm 10 a \tilde{a} (a - \tilde{a}) N^2 + 3 a \tilde{a} (2 r^2 - a \tilde{a}) N\,, \\
   & \;\;\vdots
\end{aligned}
\end{equation}
where the upper (lower) sign corresponds to $N=kr$ ($N=kr+a$).

\subsection{The large $N$ limit}

We briefly consider the large $N$ limit of these pure phase eigenvalue models. From~(\ref{eqn:singletraceEx}) we see that the contour dependence enters at $O(1/N^2)$ relative to the leading large $N$ behavior, and that there are subleading corrections suppressed by odd powers of $N$. This can be shown more generally by rewriting
\begin{equation}
f_{k,x}(q) = k \sum_{p = 0}^q \frac{(2 q - p)!}{p! (q - p)! (q - p + 1)!} (x + q - 1)^{(p)} (k - 1)^{(q - p)}\,,
\end{equation}
where $x^{(p)} \equiv x (x - 1) \ldots (x - p + 1) = \frac{\Gam{x +1}}{\Gam{x + 1 - p}}$ denotes the falling factorial. Expanding in $k \gg 1$ and retaining the first few terms, we obtain
\begin{multline} \label{eqn:singletraceN}
  r^q \langle \Tr X^{q r} \rangle = \frac{(2 q) !}{q! (q + 1) !} N^{q + 1}
  + \frac{(2 q - 2) !}{12 (q - 1) ! (q - 2) !}  (r^2  (q - 2) + 6 a\tilde{a}) N^{q - 1}\\
   \pm \frac{(2 q - 2) !}{12 (q - 1) ! (q - 3) !}  a\tilde{a} (a - \tilde{a}) N^{q - 2} + \ldots\,,
\end{multline}
where the upper (lower) sign corresponds to $N=kr$ ($N=kr+a$), as above.

Similar results hold for the free energy. We first rewrite the partition function~(\ref{eqn:ZNra}) in terms of the Barnes $G$-function, defined by the Weierstrass product
\begin{equation}
G(z+1) = e^{-\zeta'(-1)-\frac{z+(1+\gamma) z^2}{2}} \prod_{k=1}^\infty \biggl(1+\frac{z}{k}\biggr)^k e^{\frac{z^2}{2k}-z} \,,
\end{equation}
which can be shown to satisfy\footnote{A more common but ultimately less convenient convention is $G_2(z+1) = (2\pi)^{\frac{z}{2}} e^{\zeta'(-1)} G(z+1)$, which satisfies $G_2(z+1) = \Gam{z} G_2(z)$ and $G_2(1)=1$.}
\begin{align}
G(z+1) &= \frac{\Gam{z}}{\sqrt{2 \pi}} G(z) \,, & G(1) &= e^{-\zeta'(-1)}\,.
\end{align}
where $\zeta(s)$ is the Riemann zeta function and the normalization is chosen for future convenience.
We define
\begin{equation}
\mathcal{Z}_{r,a}(N) \equiv G\biggl(\frac{N}{r}+1\biggr)^r G\biggl(\frac{N+a}{r}\biggr)^a G\biggl(\frac{N+a}{r}+1\biggr)^{r-a} \,,
\end{equation}
so that
\begin{align}
Z_{N=kr}^{(r,a)} &= \delta_{r,a}(N) \frac{\mathcal{Z}_{r,a}(N)}{\mathcal{Z}_{r,a}(0)}  \,, &
Z_{N=kr+a}^{(r,a)} &= \delta_{r,a}(N) \frac{\mathcal{Z}_{r,\tilde{a}}(N)}{\mathcal{Z}_{r,a}(0)} \,.
\end{align}
We have the asymptotic expansion
\begin{equation}
\log G(n+1) = \frac{n^2}{2} \log n-\frac{3}{4}n^2 - \frac{1}{12} \log n+ \sum_{g=2}^\infty \frac{B_{2g}}{2g (2g-2)} n^{2-2g} \,,
\end{equation}
where $B_{2g}$ are the Bernoulli numbers. Thus,
\begin{equation}
\log \mathcal{Z}_{r,a}(N) = \biggl(\frac{N^2}{r} + \frac{a \tilde{a}}{2 r} - \frac{r}{6}\biggr) \log \frac{N}{r} - \frac{3 N^2}{2 r} + \frac{a \tilde{a} (a-\tilde{a})}{6 N r} - \frac{r^4 - 15 a^2 \tilde{a}^2}{120 r N^2} + \ldots \,,
\end{equation}
from which we obtain the large $N$ free energy:\footnote{Notice that the free energy satisfies $F_N^{(r,a)} = F_{N'}^{(r,r-a)}$ and $F_N^{(r,a)} = F_{-N}^{(r,r-a)}$, similar to~(\ref{eqn:schurRelns}).}
\begin{equation} \label{eqn:freeEnergyN}
F  =  F_0 + \biggl(\frac{N^2}{r} + \frac{a \tilde{a}}{2 r} - \frac{r}{6}\biggr) \log \frac{N}{r} - \frac{3 N^2}{2 r} \pm \frac{a \tilde{a} (a-\tilde{a})}{6 N r} - \frac{r^4 - 15 a^2 \tilde{a}^2}{120 r N^2} + \ldots
\end{equation}
where $F_0$ is an $N$-independent constant\footnote{We drop the prefactor $\delta_{r,a}(N)$ from the large $N$ expansion, as it is periodic in $N$ with period $2 r$, hence formally non-perturbative in $N$.} which depends on the normalization of the partition function and the the upper (lower) sign corresponds to $N=kr$ ($N=kr+a$). Ignoring the logs, the contour dependence again enters at $O(1/N^2)$ relative to the leading terms, and there are $O(1/N^3)$ corrections to the free energy.

Reproducing these results with a large $N$ analysis along the lines of~\cite{Dijkgraaf:2002fc,Dijkgraaf:2002vw,Dijkgraaf:2002dh} (see, e.g.,~\cite{Marino:2004eq} and references therein for a more comprehensive review of large $N$ techniques) is an interesting open problem. This calculation is non-trivial for several reasons. Firstly, the large $N$ analysis of~\cite{Dijkgraaf:2002fc,Dijkgraaf:2002vw,Dijkgraaf:2002dh} is naturally expressed in a contour basis of Lefschetz thimbles, corresponding to the saddle points of the potential. This basis is degenerate when the potential is monomial for $r>2$, as $r-1$ critical points coincide at the origin. Secondly, the genus expansion is organized in even powers of $N$, hence the appearance of $O(1/N^3)$ corrections is unexpected in a standard analysis.

These two issues may be linked. To pick out the pure-phase contour $C_{r,a}$, one can deform the potential by $\varepsilon \Tr X^2$ to resolve the $r-1$ critical points and then express $C_{r,a}$ in a thimble basis, later taking $\varepsilon \to 0$. The change of basis between $C_{r,a}$ and the thimbles gives a linear combination of the saddle points weighted by binomial coefficients (e.g., in the case $r=3$). The $O(N)$ terms in the sum over saddle points can change the large $N$ scaling, and we hypothesize that this gives rise to the unexpected odd powers of $N$.

In the language of topological recursion (see \cite{BEO} for a review), the appearance of subleading corrections suppressed by odd powers of $N$ means that these theories are not of \emph{topological type} (\cite{BEO}, Definition 3.6).
  Hence they are not described by the standard large-$N$ tools---spectral curve topological recursion \cite{EO,BE}---or even its more general ``blobbed'' version \cite{BS}.
 However, the difference from the topological type case is actually rather mild.
  The ward identities~(\ref{eqn:loopequation}) are not broken---they still admit solutions of topological type
  (in contrast to $\beta$-ensembles \cite{CEM}, where for generic $\beta$ solutions of topological type are forbidden).
  Rather, the unusual dependence on $N$ enters through initial conditions; for example, for $r=3$ the first non-topological correlator is $\langle (\Tr X)^3 \rangle = \pm N$.

  The large-$N$ behavior of these monomial matrix models is interesting for the following reason.
  If one computes the standard spectral curve (forgetting for now that pure phase correlators
  are not described by it), one gets $y \sim x^r$, which is the symplectic dual ($x \leftrightarrow y$)
  \cite{EO2} to the spectal curve for the r-Gelfand-Dickey hierarchy
  (see \cite{DBNOPS} Theorem 7.3).
  Since pure phases are very natural from the matrix model point of view
  one can expect that the relevant generalization of the topological recursion
  is also natural. Once available, it would immediately provide
  a generalization of the r-Gelfand-Dickey hierarchy
  (and, via a lift to cohomology, of Witten's r-spin class \cite{FSZ}).
  We defer this problem to future work.

\section{A General Orbifold Construction for Matrix Models} \label{sec:genconstruction}

The exact solutions found in the previous section arise from a more general construction, which we now explain. This construction is, roughly speaking, an orbifold, where a matrix model with a single-trace potential $W(X)$ is replaced by one with a single-trace potential $W(X^r)$.

\subsection{General results} \label{subsec:orbifoldresults}

\begin{definition} \label{defn:rfoldmodel}
For any one-matrix model
\be \label{eqn:parentmodel}
Z_{n, u} = \frac{1}{n!} \int_{C^n} \prod_{i = 1}^n \frac{d x_i}{2 \pi} x_i^u \prod_{i < j} (x_i
  - x_j)^2 e^{- \Tr W(X)} \,,
\ee
the (pure phase) \emph{$r$-fold matrix model} is defined as
\be
Z_{N, U}^{(r,a)} = \frac{1}{N!} \int_{C_{r,a}^N} \prod_{i = 1}^N \frac{d x_i}{2 \pi} (x_i^r)^\frac{U}{r} \prod_{i < j} (x_i
  - x_j)^2 e^{- \Tr W(X^r)} \,,
\ee
for any $r\in\mathbb{N}$ and choice of contour $0 \le a < r$, with
$C_{r, a} = \sum_{j = 0}^{r - 1} \omega_r^{- j a} C^{1 / r}_j$,
where $C^{1/r}$ is the principal $r$th root of $C$ and $C^{1/r}_j$ is $C^{1/r}$ rotated by $\omega_r^j$.
\end{definition}
 Here $W(X)$ is any (single-trace) potential and $C$ is any contour,\footnote{Note that $C^{1/r}$ will have additional discontinuities relative to $C$ if $C$ crosses the negative real axis, since the principal $r$th root is discontinuous there. The pieces of the contour are reconnected in the linear combination $C_{r,a}$, but the weight of each piece remains discontinuous except in the special case $a=0$.} where we isolate a contribution $-u \log X$ from the potential for future convenience.\footnote{The presence of the $u$ deformation is linked to the inclusion of the contour $a=0$. The additional logarithmic term in the potential adds a saddle point, hence (in the basis of Lefschetz thimbles) an additional integration contour with a boundary at the origin.} We assume that the solution to the parent matrix model is known, and use it to solve the corresponding $r$-fold model.

Let $p_m^{(u)}$ be a set of orthogonal polynomials for the parent model~(\ref{eqn:parentmodel}):
\begin{equation} \label{eqn:parentortho}
  Z_{1, u} [p_m^{(u)} p_n^{(u)}] = t_m^{(u)} \delta_{m, n} \,.
\end{equation}
Now consider the polynomials
\begin{equation}
  P_{k r + i}^{(U ; r, a)} (x) = \begin{cases}
    x^i p_k^{\left(\frac{U+a}{r} - 1 \right)} (x^r)\,, & 0 \leqslant i <
    a\,,\\
    x^i p_k^{\left(\frac{U+a}{r} \right)} (x^r)\,, & a \leqslant i < r \,.
  \end{cases}
\end{equation}
Using~(\ref{eqn:parentortho}), one can check that
\begin{align}  \label{eqn:rfoldortho}
  Z_{1, U}^{(r, a)} [P_m^{(U ; r, a)} P_n^{(U ; r, a)}] &= T_m^{(U ; r, a)}
  \delta_{m, n}^{(r, a)} \,, & T_{k r + i}^{(U ; r, a)} &= \begin{cases}
    t_k^{\left( \frac{U+a}{r} - 1 \right)}\,, & 0 \leqslant i < a\,,\\
    t_k^{\left( \frac{U+a}{r} \right)} \,, & a \leqslant i < r\,.
  \end{cases}
\end{align}
where $\delta_{m, n}^{(r, a)}$ is the same as in~(\ref{eqn:orthocond}).
In particular, (\ref{eqn:rfoldortho}) follows from the $\mathbb{Z}_r$ orbifold projection implied by the contour $C_{r,a}$ together with the change of variables
\be
r \int_{C^{1 / r}} x^{a-1} f(x^r)\, d x = \int_{C} y^{\frac{a}{r}-1} f(y)\, d y \,,
\ee
where $(x^r)^{1/r} = x$ for $-\frac{\pi}{r} <\arg x \le \frac{\pi}{r}$.

Using~(\ref{eqn:rfoldortho}), we can compute the partition function $Z_{N,U}^{(r,a)}$, much as in~\S\ref{subsec:orthopolys}:
\begin{equation}
  Z_{N, U}^{(r, a)} = \delta_{r, a} (N) \prod_{I = 0}^{N - 1} t_{\floor{\frac{I}{r}}}^{\left( \frac{U+a}{r} + \floor{\frac{I - a}{r}} -\floor{\frac{I}{r}} \right)} \,.
\end{equation}
Re-expressing this in terms of the partition function of the parent model, $Z_{n,u} = \prod_{i=0}^{n-1} t_i^{(u)}$, we find\footnote{Here we use the substitution $I = r i+(N-\mu-1) \bmod r$ for $i\ge 0$ and $0 \le \mu < r$. Simpler substitutions are possible, such as $I=r i +\mu$, but this particular form occurs naturally in the Schur average formula.}
\begin{theorem} \label{thm:genZFormula}
The partition function of an $r$-fold matrix model is a product of $r$ copies of the partition function of the parent model:
\begin{equation} \label{eqn:genZFormula}
  \boxed{Z_{N, U}^{(r, a)} = \delta_{r, a} (N) \prod_{\mu = 0}^{r - 1} Z_{n_{\mu},
  u_{\mu}}\,,}
\end{equation}
where
\begin{align} \label{eqn:nmuumu}
  n_{\mu} &= \floor{\frac{N - \mu - 1}{r}} + 1 \,, &
   u_{\mu} &= \frac{U+a}{r} + \floor{\frac{N - \mu - a - 1}{r}} - \floor{\frac{N - \mu - 1}{r}} \,.
\end{align}
\end{theorem}
Note that $\sum_{\mu} n_{\mu} = N$ and $\sum_{\mu} u_{\mu} = U$.

A similar factorized structure occurs in the Schur average formula. To derive it, we start with the general result~(\ref{eqn:polyschurformula}). Writing $p_n^{(u)}(x) = \sum_i p_{n; i}^{(u)}\, x^i$, we have
\be
 P_N^{(U ; r, a)} (x) = \sum_I \delta_{r| (N - I)}\, p_{\floor{\frac{N}{r}}; \floor{\frac{I}{r}}}^{\left( \frac{U+a}{r} + \floor{\frac{N - a}{r}} - \floor{\frac{N}{r}} \right)} x^I \,,
\ee
and so
\begin{equation}
  \langle s_{\lambda} (X) \rangle^{(r,a)}_{N,U} = (- 1)^{| \lambda |} \underset{1 \leqslant
  I, J \leqslant K}{\det}  \left[ \delta_{r| (\lambda_J^\T + I - J)}\, p_{\floor{\frac{N + I - 1}{r}}; \floor{ \frac{N + J - 1 - \lambda_J^\T}{r}}}^{\left( \frac{U+a}{r} + \floor{\frac{N + I - 1 - a}{r}} - \floor{\frac{N + I - 1}{r}} \right)} \right] \,, \qquad K \ge \lambda_1 \,.
\end{equation}
Choosing $K$ such that $K+N \equiv 0 \pmod r$ and applying~(\ref{eqn:genDetFormulaTranspose}), we obtain
\be
\langle s_{\lambda} (X) \rangle^{(r,a)}_{N,U} = \delta_r(\lambda) \prod_{\mu = 0}^{r-1} (-1)^{|\lambda/r^{(\mu)}|} \det_{1\le i,j \le k_\mu} p^{(u_\mu)}_{n_{\mu} + i -1;\, n_{\mu} + j - (\lambda/r^{(\mu)})_j^\T-1} \,.
\ee
where $k_\mu = \floor{\frac{K+\mu}{r}}$. From the definition of the $r$-quotient~(\ref{eqn:pquotient}), we find
\be
k_\mu \ge \floor{\frac{\lambda_1+\mu}{r}} \ge (\lambda/r^{(\mu)})_1 \,,
\ee
so that
\begin{theorem} \label{thm:genSchurFormula}
The average of Schur polynomials in an $r$-fold matrix model is the product of $r$ averages in the parent model:
\be \label{eqn:genSchurFormula}
\boxed{\langle s_{\lambda} (X) \rangle_{N, U}^{(r, a)} = \delta_r (\lambda)
  \prod_{\mu = 0}^{r - 1} \langle s_{\lambda / r^{(\mu)}} (X)
  \rangle_{n_{\mu}, u_{\mu}} \,, }
\ee
where $n_\mu$ and $u_\mu$ are defined in Theorem~\ref{thm:genZFormula}.
\end{theorem}
An example application of this theorem is shown in Figure~\ref{fig:quotientCorr}.
\begin{figure}
  \centering
  \includegraphics[width=0.75\textwidth]{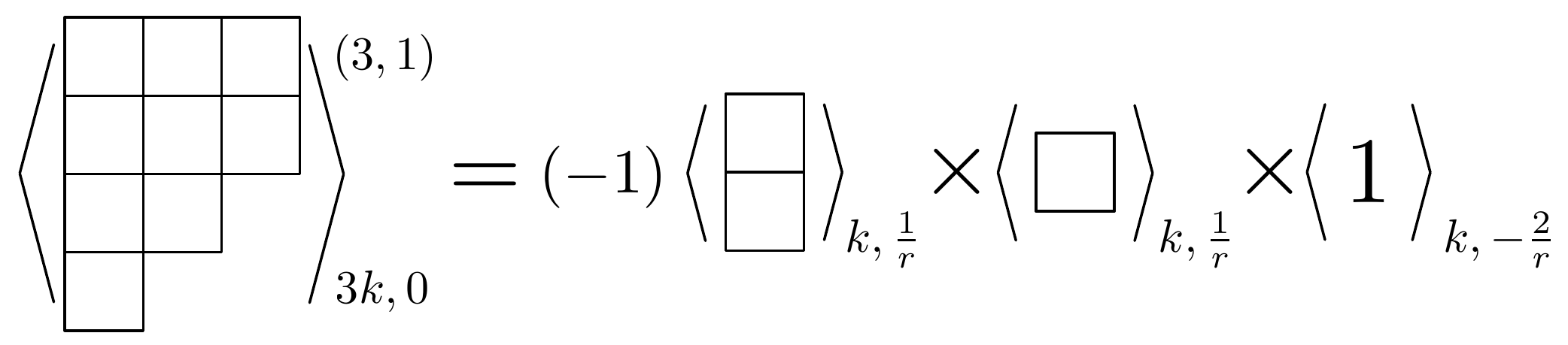}
  \caption{One example of applying Theorem~\ref{thm:genSchurFormula}, c.f.\ Figure~\ref{fig:3foldCorr}.\label{fig:quotientCorr}}
\end{figure}

\subsection{Example: logarithmic models} \label{sec:logmodels}

As an example of the general construction given in the previous section, we consider the eigenvalue model
\begin{equation} \label{eqn:logZmodel}
Z_N^{(r,a)} = \frac{1}{N!} \int_{C_{r,a}^N} \prod_i \frac{\dd x_i}{2 \pi}  \prod_{i < j} (x_i - x_j)^2  \prod_i (x_i^r)^u  (1 - x_i^r)^v \,,
\end{equation}
where
\begin{equation} \label{eqn:log-contours}
C_{r,a} = \sum_{j = 0}^{r - 1} \omega_r^{- j a} B_j \,,
\end{equation}
\smallskip\\
with $B_j$ now the finite segment $[0, \omega_r^j]$. This is an $r$-fold generalization of the $\beta = 1$ Selberg integral~\cite{Selberg:1941,Selberg:1944} (see also~\cite{Forrester:2007}):\footnote{One might hope to generalize the $\beta$-deformation $(x_i - x_j)^2 \to |x_i - x_j|^{2 \beta}$ to the exactly solvable $r>1$ models we study in this paper. However, this is challenging due to the non-analytic integrand that results for non-integral $\beta$, and even for integral $\beta$ we do not find any exact results besides the $r=2$, $a=1$ case with $u=0$, which is merely a special case of the Selberg integral.}
\begin{equation} \label{eqn:selberg}
\frac{1}{N!} \int_0^1 \prod_i \dd x_i  \prod_{i < j} (x_i - x_j)^2
  \prod_i x_i^u  (1 - x_i)^v = \prod_{i = 1}^{N} \frac{\Gamma (i + u) \Gamma (i + v) \Gamma (i)}{\Gamma (N + i + u + v)} \,.
\end{equation}
The model~(\ref{eqn:logZmodel}) is related to the exponential model~(\ref{eqn:modeldef}) by the scaling limit
\begin{equation}
\lim\limits_{v\to \infty} \biggl(1-\frac{1}{v} x^r \biggr)^v = e^{-x^r} \,,
\end{equation}
hence $\lim_{v\to \infty} S_{0,v}[X/v^{1/r}] = \Tr X^r$, where $S_{u,v}[X]$ is the logarithmic potential
\begin{equation} \label{eqn:Suv}
S_{u,v}[X] = -u \Tr \ln \big( X^r \big) - v \Tr \ln \big(1 - X^r\big) \,.
\end{equation}
This logarithmic generalization is natural, for example, from the perspective of conformal field theories, where free-field representations of correlation functions are typically given by generalized matrix models with logarithmic potentials \cite{Dotsenko:1984nm} sometimes referred to as Selberg integrals \cite{SelbergReview}.

Using Theorem~\ref{thm:genZFormula} and~(\ref{eqn:selberg}), we obtain the partition function
\begin{equation} \label{eqn:logZ}
\boxed{ \ \ Z_N^{(r,a)} = \frac{\delta_{r,a}(N)}{(2 \pi)^N}  \prod_{I =
  0}^{N - 1}  \frac{\Gam{\left\lfloor \frac{I - a}{r}
  \right\rfloor + \frac{a}{r} + u + 1} \Gam{\left\lfloor
  \frac{I}{r} \right\rfloor + v + 1} \Gam{\left\lfloor
  \frac{I}{r} \right\rfloor + 1}}{\Gam{\left\lfloor \frac{N+I - a}{r}
  \right\rfloor + \frac{a}{r} + u + v + 1}} \,,\ \ }
\end{equation}
where we use the substitution
\be
I = \begin{cases} i r - \mu - 1\,, & N=kr\,,\\ (i-1) r + \mu\,, & N = k r+a\,, \end{cases}
\ee
to simplify the product of $r$ Selberg integrals.

The expectation value of a Schur polynomial with the Selberg measure~(\ref{eqn:selberg}) is given by the Kadell formula~\cite{Kadell:1988}:
\begin{equation}
\langle s_{\lambda} (X) \rangle = \prod_{x \in \lambda} \frac{(N + c_\lambda(x)) (N + c_\lambda(x) + u)}{h_\lambda(x) (2N + c_\lambda(x)+ u+v)}\,.
\end{equation}
Using this together with Theorem~\ref{thm:genSchurFormula}, we can compute the average of a Schur polynomial in the $r$-fold model~(\ref{eqn:logZmodel}). To simplify the result, we rewrite the product over contents in the $r$-quotient as a product of columns using $\prod_{x \in \lambda} \Gamma(\alpha+c_\lambda(x)) = \prod_{j=1}^{\lambda_1} \frac{\Gamma(\alpha+j)}{\Gamma(\alpha+j-\lambda^\T_j)}$. Rearranging and applying a generalization of~(\ref{eqn:modfactorial}), the result can be expressed as a product over boxes of the parent Young diagram:
\begin{equation} \label{eqn:log1schur}
\boxed{\ \  \langle s_{\lambda} (X) \rangle = \delta_r(\lambda) \prod_{x \in \lambda} \frac{\llbracket N + c_\lambda(x) \rrbracket_{r,0}\, \llbracket N + c_\lambda(x); ru \rrbracket_{r,a}}{\llbracket h_\lambda(x) \rrbracket_{r,0} \, \llbracket 2N + c_\lambda(x); ru+rv \rrbracket_{r,a}}\,,\ \ }
\end{equation}
where the deformed bracket $\llbracket n; x\rrbracket_{r,a}$ generalizes~(\ref{eqn:bracketdef1})
\begin{equation} \label{eqn:bracketdef2}
\llbracket n; x \rrbracket_{r, a} = \begin{cases} n+x & n \equiv a \mod r \\ 1 & \text{otherwise} \end{cases} \,.
\end{equation}
It is straightforward to check that~(\ref{eqn:logZ}) and~(\ref{eqn:log1schur}) reduce to~(\ref{eqn:ZNra}) and~(\ref{eqn:1schur}) in the appropriate limit.

 These logarithmic models
 not only preserve all the solvability properties discussed so far; they also reveal additional ones not present in the original polynomial models. They enjoy exact formulas for the correlation function of two Schur polynomials (with an appropriate shift in the argument). This remarkable property has already been noted for $r=1$ models \cite{Kadell:1993, SelbergReview, Mironov:2010pi, Alba:2010qc}; in this paper we propose a generalization to $r > 1$, see \S\ref{subsec:2schur}.

\subsection{Some $q$ analogs} \label{subsec:qanalogs}

The orbifold-like construction of~\S\ref{subsec:orbifoldresults} can also be extended to $q$-analogs of random matrix models. $q$ calculus is based on the replacement
\be
\frac{d f}{d x} \longrightarrow \frac{d_q f}{d_q x} \equiv \frac{f(x) - f(q x)}{x- q x} \,,
\ee
for some $0<q<1$, where the limit $q\to 1$ takes $\frac{d_q f}{d_q x} \to \frac{d f}{d x}$. The inverse of the $q$ derivative is the Jackson integral
\be
\int f(x) d_q x \equiv (1-q) x \sum_{k=0}^\infty q^k f(q^k x) \,,
\ee
where definite integrals are defined by the fundamental theorem of calculus, $\int_a^b f(x) d_q x = F(b) - F(a)$ for $F(x) = \int f(x) d_q x$.

$q$ calculus satisfies a limited form of the chain rule
\be
d_{q^r} (a x^r) = a [r]_q x^{r-1} d_q x \,,
\ee
where $[n]_q$ denotes the $q$-number
\be
[n]_q \equiv \frac{1-q^n}{1-q} \,,
\ee
which satisfies $\lim_{q\to 1} [n]_q = n$. Such calculus appears naturally from several closely related perspectives. Physically, it has most recently attracted attention in the context of five-dimensional gauge theories, where the finite difference parameter $q$ is the exponentiated radius of the fifth dimension \cite{Nekrasov:2003af} in the spirit of Kaluza and Klein. Mathematically, $q$-numbers appear in enumerative geometry of symplectic resolutions as K-theory weights associated with fixed points of equivariant torus action \cite{Okounkov:2015spn}. From the perspective of integrability theory, $q$-numbers correspond to trigonometric integrable models, which occupy an intermediate level of complexity between rational (corresponding to usual numbers) and elliptic (corresponding to elliptic numbers, \cite{Frenkel1997}) models.

The reasoning of~\S\ref{subsec:orbifoldresults} now goes as follows.
\begin{definition}
For any $q$-deformed one-matrix model
\be \label{eqn:Qparentmodel}
Z_{n, u; q} = \frac{1}{n!} \int_{C^n} \prod_{i = 1}^n \frac{d_q x_i}{2 \pi} x_i^u \prod_{i < j} (x_i
  - x_j)^2 e^{- \Tr W(X;q)} \,,
\ee
the (pure phase) \emph{$r$-fold matrix model} is defined as
\be
Z_{N, U; q}^{(r,a)} = \frac{1}{N!} \int_{C^N_{r,a}} \prod_{i = 1}^N \frac{d_q x_i}{2 \pi} (x_i^r)^{\frac{U}{r}} \prod_{i < j} (x_i
  - x_j)^2 e^{- \Tr W(X^r;q^r)} \,,
\ee
with $C_{r,a}$ defined as in Definition~\ref{defn:rfoldmodel}.
\end{definition}
This definition is chosen so that the $q$-deformed measures of the parent and $r$-fold models are related by a change of variables. 
 We find the $r$-fold orthogonal polynomials
\begin{equation}
  P_{k r + i}^{(U ; r, a)} (x;q) = \begin{cases}
    x^i p_k^{\left(\frac{U+a}{r} - 1 \right)} (x^r;q^r)\,, & 0 \leqslant i <
    a\,,\\
    x^i p_k^{\left(\frac{U+a}{r} \right)} (x^r;q^r)\,, & a \leqslant i < r \,,
  \end{cases}
\end{equation}
with
\begin{align}
  Z_{1, U;q}^{(r, a)} [P_m^{(U ; r, a)} P_n^{(U ; r, a)}] &= T_{m;q}^{(U ; r, a)}
  \delta_{m, n}^{(r, a)} \,, & T_{k r + i;q}^{(U ; r, a)} &= \frac{r}{[r]_q} \cdot \begin{cases}
    t_{k;q^r}^{\left( \frac{U+a}{r} - 1 \right)}\,, & 0 \leqslant i < a\,,\\
    t_{k;q^r}^{\left( \frac{U+a}{r} \right)} \,, & a \leqslant i < r\,.
  \end{cases}
\end{align}
By the same reasoning as above, we conclude that
\begin{theorem} \label{thm:genQanalog}
The partition function and Schur averages of a $q$-deformed $r$-fold matrix model factor:
\be \label{eqn:genQanalog}
  \boxed{\begin{aligned}Z_{N, U;q}^{(r, a)} &= \delta_{r, a} (N) \left(\frac{r}{[r]_q}\right)^N\, \prod_{\mu = 0}^{r - 1} Z_{n_{\mu},
  u_{\mu};q^r}\,, & \langle s_{\lambda} (X) \rangle_{N, U;q}^{(r, a)} &= \delta_r (\lambda)
  \prod_{\mu = 0}^{r - 1} \langle s_{\lambda / r^{(\mu)}} (X)
  \rangle_{n_{\mu}, u_{\mu};q^r} \,, \end{aligned}}
\ee
as in Theorems~\ref{thm:genZFormula} and~\ref{thm:genSchurFormula}.
\end{theorem}

As an example, we consider the $q$-Selberg integral
\begin{equation} \label{eqn:qselberg}
\frac{1}{N!} \int_0^1 \prod_i \dd_q x_i  \prod_{i < j} (x_i - x_j)^2
  \prod_i x_i^u  (q x_i; q)_v = q^{(u+1) \binom{N}{2}+2 \binom{N}{3}  } \prod_{i = 1}^{N} \frac{\Gamma_q (i + u) \Gamma_q (i + v) \Gamma_q (i)}{\Gamma_q (N + i + u + v)} \,.
\end{equation}
Here
\begin{align}
(x;q)_\infty &\equiv \prod_{i=0}^{\infty} (1-q^i x) \,, & (x;q)_n &\equiv \frac{(x;q)_\infty}{(x q^n;q)_\infty}\,,
\end{align}
is the $q$-Pochhammer symbol and $\Gamma_q(x) \equiv (1-q)^{1-x} (q;q)_{x-1}$ is the $q$-gamma function, which satisfies $\Gamma_q(x+1) = [x]_q \Gamma_q(x)$, $\Gamma_q(1) = 1$, and $\lim_{q\to1} \Gamma_q(x) = \Gamma(x)$.

The $r$-fold generalization of the $q$-Selberg integral is the eigenvalue model:
\be
Z_{N}^{(r,a)} = \frac{1}{N!} \int_{C^N_{r,a}} \prod_{i = 1}^N \frac{\dd_q x_i}{2 \pi} (x_i^r)^{u} (q^r x_i^r; q^r)_{v} \prod_{i < j} (x_i
  - x_j)^2  \,.
 \ee
Using~(\ref{eqn:genQanalog}), we obtain
\begin{multline}
 Z_N^{(r,a)} = \frac{\delta_{r,a}(N)}{(2 \pi)^N}  (q^r)^{n_r(N,u)} \left(\frac{r}{[r]_q}\right)^N\, \\
 \times  \prod_{I = 0}^{N - 1}  \frac{\Gamma_{q^r}\bigl(\floor{\frac{I - a}{r}} + \frac{a}{r}+u + 1\bigr) \Gamma_{q^r}\bigl(\floor{\frac{I}{r}} + v + 1\bigr) \Gamma_{q^r}\bigl(\floor{\frac{I}{r}} + 1\bigr)}{\Gamma_{q^r}\bigl(\floor{\frac{N+I - a}{r} } + \frac{a}{r}+u+v + 1\bigr)} \,,
\end{multline}
where $n_r(N,u)\equiv \sum_{\mu = 0}^{r-1} \left[(u_\mu+1) \binom{n_\mu}{2}+2 \binom{n_\mu}{3}\right]$.\footnote{One can check that $n_r(N,u) = r (u+1) \binom{\floor{N/r}}{2}+2 r \binom{\floor{N/r}}{3}+\floor{N/r} (N \bmod r) (u+N/r-1)$, hence it does not depend separately on $a$.}
Apart from the $q$-dependent prefactor, this is a straightforward $q$-analog of~(\ref{eqn:logZ}).

Likewise, using the $q$-Selberg Schur average formula~\cite{Kadell:1988}
\be
\langle s_{\lambda} (X) \rangle = q^{n(\lambda)} \prod_{x \in \lambda} \frac{[N + c_\lambda(x)]_q [N + c_\lambda(x) + u]_q}{[h_\lambda(x)]_q [2N + c_\lambda(x)+ u+v]_q} \,,
\ee
with $n(\lambda) \equiv \sum_i (i-1)\lambda_i$, we obtain
\begin{equation}
\langle s_{\lambda} (X) \rangle^{(r,a)} = \delta_r(\lambda)\, (q^r)^{n_r(\lambda)} \prod_{x \in \lambda} \frac{\llbracket N + c_\lambda(x) \rrbracket_{r,0}^{(q)}\, \llbracket N + c_\lambda(x); ru \rrbracket_{r,a}^{(q)}}{\llbracket h_\lambda(x) \rrbracket_{r,0}^{(q)} \, \llbracket 2N + c_\lambda(x); ru+rv \rrbracket_{r,a}^{(q)}}\,,
\end{equation}
where $\llbracket n; x \rrbracket_{r,a}^{(q)}$ is given by the obvious $q$-deformation $n+x \to [n+x]_q$ of~(\ref{eqn:bracketdef2}) and $n_r(\lambda) \equiv \sum_\mu n(\lambda/r^{(\mu)})$. This is a $q$-analog of~(\ref{eqn:log1schur}).

\paragraph{The $q \to \omega_r$ limit}

There is a another class of $q$-analogs of the $r$-fold models discussed in this paper. To see this, consider a Jackson integral
\be
\int_{C} e^{-W(x;q)} d_q x \,,
\ee
where for simplicity we take $C$ to be an interval along the positive real axis. Whereas in the limit $q\to 1$ we obtain an ordinary integral along $C$, in the limit $q \to \omega_r$, we find
\be
\lim_{q \to \omega_r} [r]_q \int_a^b e^{-W(x;q)} \dd_q x = \int_{C_{r,0}} e^{-W_r(x)} \dd x\,,
\ee
where $W_r(x) = W(x;\omega_r)$ and $C_{r,0}$ is the $a=0$ $r$-fold contour of Definition~\ref{defn:rfoldmodel}. Other $r$-fold contours can be obtained by inserting appropriate branch cuts (in the form $x^{-a} (x^r)^{a/r}$) into the integrand.

Thus, for a given $r$-fold potential $W(x^r)$, any $q$-deformed potential $W(x;q)$ satisfying $W(x;\omega_r)=W(x^r)$ defines a $q$-analog which reduces to the $r$-fold model in the $q \to \omega_r$ limit. For example, in the $q$-Selberg measure we have
\be
\lim_{q \to \omega_r} (q x_i; q)_V = (1-x_i^r)^{V/r} \,,
\ee
provided that $V \in r \mathbb{Z}$. Thus, the $r=1$ $q$-Selberg integral is in some sense a $q$-analog of the $r$-fold logarithmic model discussed in~\S\ref{sec:logmodels}. Using this connection and analytic continuation off the integers, one obtains an alternate proof of~(\ref{eqn:logZ}) and~(\ref{eqn:log1schur}) from the $q$-Selberg integral.

The limit $q\to \omega_r$ has an interesting connection to orbifolds (see, e.g., \cite{Kimura:2011zf}), which plays a role in several possible physical applications for $r$-fold matrix models, as discussed in~\S\ref{sec:future}.

\subsection{Refactorization of the Vandermonde} \label{subsec:refactorization}

The results of~\S\ref{subsec:orbifoldresults} can be reduced to a combinatoric identity of the integrand, as follows. Define the projection operator
\be
P_a f(x) = \frac{1}{r} \sum_{j=0}^{r-1} \omega_r^{-j a} f(\omega_r^j x) \,,
\ee
and notate
\be
P_{a_1, \ldots, a_N} f(x_1,\ldots, x_N) = P_{a_1}^{(x_1)} \ldots P_{a_N}^{(x_N)} f(x_1,\ldots, x_N) \,.
\ee
We will show that
\begin{theorem} \label{thm:VanProjection}
 The projection of the Vandermonde factors
\begin{equation} \label{eqn:VanProjection}
\boxed{
P_{a_1, \ldots, a_N} \prod_{I>J} (x_I - x_J) = \delta_{a_1,\ldots,a_N} \prod_I x_I^{a_I \bmod r} \prod_{\mu=0}^{r-1} \prod_{i>j} (x_{I_{\mu,i}}^r-x_{I_{\mu,j}}^r) \,,}
\end{equation}
where $I_{\mu,i}$ denotes the $i$th value of $I$ for which $a_I \equiv \mu \pmod r$ and $\delta_{a_1, \ldots, a_N} \in \{0,\pm1\}$ is defined as follows: consider the function $\sigma: [1,N] \to \mathbb{N}$
\be
\sigma(I_{\mu,i}) = 1+r (i-1) + \mu \,.
\ee
If $\sigma$ is a permutation on $[1,N]$, then $\delta_{a_1, \ldots, a_N}=(-1)^\sigma$. Otherwise $\delta_{a_1, \ldots, a_N}=0$.
\end{theorem}

\begin{proof}
Call the polynomial on the left hand side $f(x_1,\ldots,x_N)$. The homogeneity conditions $f(\ldots,\omega_r x_I,\ldots) = \omega_r^{a_I} f(\ldots, x_I,\ldots)$ imply that $f(x_1,\ldots,x_N) = \bigl(\prod_I x_I^{a_I \bmod r}\bigr) g(x_1^r, \ldots, x_N^r)$ for some polynomial $g(y_1, \ldots, y_N)$. Since $f$ is alternating within each sub-block $I_{\mu, i}$ for fixed $\mu$, so too is $g$, hence $g$ contains Vandermonde factors, and~(\ref{eqn:VanProjection}) holds for some polynomial $\delta = \delta(x_1^r, \ldots, x_N^r)$. The total degree of $\delta$ in the $x$ variables is
\be \label{eqn:degreedelta}
\sum_{\mu = 0}^{r-1} \biggl(\mu N_\mu +\frac{N_\mu (N_\mu-1)}{2} \biggr) -\frac{N(N-1)}{2} \,,
\ee
where $N_\mu$ denotes the number of variables $x_I$ for which $a_I \equiv \mu \pmod r$. Maximizing~(\ref{eqn:degreedelta}) at fixed $N$, we find the conditions $N_\mu - N_\nu \in \{0,1\}$ for $\mu \le \nu$, thus the unique maximum satisfies $N_\mu = n_\mu$ from~(\ref{eqn:nmuumu}). In this case, one can check that~(\ref{eqn:degreedelta}) vanishes, implying that $\delta$ is a constant, whereas for $N_\mu \ne n_\mu$ the total degree of $\delta$ is negative, requiring $\delta = 0$.

To fix the constant $\delta$ in the case where $N_\mu = n_\mu$, note that we can fix $a_I \equiv I-1 \pmod r$ up to a permutation of the labels. Comparing the coefficients of $\prod_I x_I^{I-1}$ on both sides, we conclude that $\delta_{0,\ldots,N-1} = 1$, whereas $\sigma$ is the identity permutation. Generalizing by permuting the labels gives $\delta_{a_1, \ldots, a_N}=(-1)^\sigma$. Since moreover $\sigma$ is a permutation on $[0,N]$ iff $N_\mu = n_\mu$, the result~(\ref{eqn:VanProjection}) is proven.
\end{proof}

In fact, Theorem~\ref{thm:VanProjection} encompasses our earlier results Theorems~\ref{thm:genZFormula}--\ref{thm:genQanalog}. For instance, the pure-phase partition function can be derived by rewriting
\be \label{eqn:SNfix1}
\frac{1}{N!} \left(\det_{N\times N} x_i^{N-j}\right)^2 \cong \prod_i x_i^{N-i} \prod_{i<j} (x_i - x_j) \,,
\ee
up to symmetrization in the variables $x_i$. The integral over $C_{r,a}$ imposes the projection $P_{a-1,\ldots,a-1}$ on the integrand, hence $P_{a-N,\ldots, a-1}$ on the Vandermonde $\prod_{i<j} (x_i - x_j)$. By Theorem~\ref{thm:VanProjection}, this splits the Vandermonde into $r$ non-interacting blocks, and we recover Theorem~\ref{thm:genZFormula}. Likewise, using the definition~(\ref{eqn:Schurdef}) we find
\be \label{eqn:SNfix2}
\frac{1}{N!} s_{\lambda}(X) \left(\det_{N\times N} x_i^{N-j}\right)^2 \cong \prod_i x_i^{N+\lambda_i-i} \prod_{i<j} (x_i - x_j) \,,
\ee
up to symmetrization. As before, the projection implied by the contour splits the Vandermonde into $r$ non-interacting blocks, and we recover Theorem~\ref{thm:genSchurFormula}.

Theorem~\ref{thm:VanProjection} can also be used to answer several important questions that we will not discuss in detail in the present paper, but which deserve further attention in future work. Firstly, using this identity we can compute the unnormalized insertions $Z_N^{(r,a)}[s_{\lambda}(X)]$ in cases where $Z_N^{(r,a)}[1] = 0$ ($N \not\equiv 0, a \pmod r$). More importantly, Theorem~\ref{thm:VanProjection} can be used to derive a factorized structure for \emph{mixed phases}, as follows.

After fixing the $S_N$ permutation symmetry as in~(\ref{eqn:SNfix1}) or~(\ref{eqn:SNfix2}), a mixed phase integrated on the contour $\prod_{a=0}^{r-1} C_{r,a}^{N_a}$ ($\sum_a N_a = N$) becomes an integral over the symmetrization of the contour $\prod_a C_{r,a}^{N_a}$. This symmetrization can be broken up into ordered components $\prod_{i=1}^N C_{r,a_i}$, such that $x_i$ is integrated along the contour $C_{r,a_i}$. This is a finer basis of contours than that discussed in~\S\ref{sec:puremixed} because in a mixed phase each symmetrized contour has multiple ordered components $\prod_{i=1}^N C_{r,a_i}$. In particular, for $p$ eigenvalue contours there are $p^N$ ordered contours but only $\binom{N+p-1}{p-1}$ symmetrized contours. The larger basis of contours does not lead to further ambiguities in the loop equations (recall~\S\ref{sec:contourdependence}) because the weights of ordered contours related by permutations are constrained to be equal. Nonetheless, the basis of ordered contours is useful because on each ordered contour the eigenvalue interactions~(\ref{eqn:SNfix1}) and~(\ref{eqn:SNfix2}) factor into the product of $r$ subblocks using Theorem~\ref{thm:VanProjection}, much as in Theorems~\ref{thm:genZFormula} and~\ref{thm:genSchurFormula}.

This means, for instance, that the partition function $Z$ of the reflection-positive quartic matrix model discussed in~\S\ref{sec:quartic} can be written explicitly as a sum of $2^N$ terms (since $(C_{4,1}+C_{4,3})/2$ equals the real line), as can $Z[s_{\lambda}(X)]$. While the number of terms grows rapidly with $N$, the growth is much slower than the $N!$ terms which appear in the determinant~(\ref{eqn:quarticdet}).

Moreover, this approach gives a straightforward solution to any given ``symmetrized'' mixed phase as a finite sum over ``ordered'' mixed phases, where the ordered mixed phases are given by $r$-fold products of the parent theory. Whether this leads to further insights into the mixed phases (and by extension, strongly-interacting reflection-positive models) is a question for the future.

\section{Applications and Future Directions} \label{sec:future}

In this section we discuss various generalizations and applications of our results. Several open questions still remain here and clarifying those questions is a promising direction of future research.

\subsection{Two-Schur correlators} \label{subsec:2schur}

Perhaps the most puzzling and not fully understood phenomenon, observed in a wide range of matrix models, is a possibility to write an exact formula for a \emph{two-Schur correlator}: an average of a product of two Schur polynomials. While existence of  a formula for the average of a single Schur polynomial is not surprising and relates to, e.g., the determinant structure of Schur polynomials, an equally simple reason for the exact solvability of two-Schur correlators is---to the best of our knowledge---unknown.

In this section, we present such a formula for two-Schur averages in the general case of the $r$-fold $q$-deformed logarithmic model. We have tested this formula with numerous computer experiments.\footnote{The necessary MAPLE code may be found at \url{http://math.harvard.edu/~shakirov/}.}
\begin{conjecture} \label{conj:twoschur}
For a pair of partitions $\lambda, \mu$ let $Z_{\lambda \mu}$ be the following product,
\begin{equation}
\begin{split}
Z_{\lambda \mu} &= \prod\limits_{1 \leq i < j \leq N} \dfrac{\lfloor \lambda_i - \lambda_j + j - i; \ 0 \rfloor_{r,0}}{\lfloor j - i; \ 0 \rfloor_{r,0}} \ \prod\limits_{1 \leq i < j \leq N} \dfrac{\lfloor \mu_i - \mu_j + j - i; \ 0 \rfloor_{r,0}}{\lfloor j - i; \ 0 \rfloor_{r,0}} \\
&\mathrel{\phantom{=}} \prod\limits_{1 \leq i,j \leq N} \dfrac{\lfloor 2 N + 1 - i - j; \ u + v \rfloor_{r,a}}{\lfloor 2 N + 1 + \lambda_i + \mu_j - i - j; \ u+v \rfloor_{r,a}} \\
&\mathrel{\phantom{=}} \prod\limits_{(i,j) \in \lambda} \dfrac{\lfloor N + j - i; \ u \rfloor_{r,a}}{ \lfloor N + j - i; \ u + v \rfloor_{r,a} } \prod\limits_{(i,j) \in \mu} \dfrac{\lfloor N + j - i; \ v \rfloor_{r,0}}{ \lfloor N + j - i; \ 0 \rfloor_{r,0}}\,,
\label{2Schur}
\end{split}
\end{equation}
\smallskip\\
with the linear factors given by
\begin{equation}
\lfloor x; y \rfloor_{r,a} =
\begin{cases}
\Lambda \ \dfrac{q^{(x+ry)/2} - q^{-(x+ry)/2}}{q^{1/2} - q^{-1/2}}, &(x+a) \bmod r = 0\,, \\
\dfrac{\omega_r^{x/2} - \omega_r^{-x/2}}{\omega_r^{1/2} - \omega_r^{-1/2}}, &(x+a) \bmod r \neq 0\,.
\end{cases}
\end{equation}
Then,
\begin{align}
\label{eq:cong}
\left\langle \ \chi_{\lambda}\big( X \big) \ \chi_{\mu}\left( \Tr X^k \ \mapsto \ \Tr X^k + rv \ \delta_{r|k} \right) \ \right\rangle =
\begin{cases}
Z_{\lambda \mu}, & \mbox{if} \ \deg_{\Lambda} Z_{\lambda \mu} = 0\,, \\
0, & \mbox{otherwise.}
\end{cases}
\end{align}
\end{conjecture}
 Here the notation $s_{\lambda}( \Tr X^k \mapsto \Tr X^k + rv \delta_{r|k})$ means that we first write $s_\lambda$ as a linear combination of multitrace operators and then replace $\Tr X^k \mapsto \Tr X^k + rv \delta_{r|k}$.
Note that this formula does not fully survive the polynomial limit: when $x \mapsto x / v$ and $v \rightarrow \infty$,  the $v$-shift in the second Schur polynomial implies that the $\mu$-dependence goes away, and the formula reduces to a single Schur correlator.  Thus, there is no analogous two-Schur average formula in the monomial matrix models considered in~\S\ref{sec:exactsolutions}.

The $r=1$ case of Conjecture~\ref{conj:twoschur} was given in \cite{Mironov:2010pi, Alba:2010qc,Mironov:2011dk} and used to formulate a proof of the AGT conjecture. For $v = 0$ these conjectures reduce to a theorem proven by Kadell~\cite{Kadell:1993}. The space of conjectured solvable two-Schur correlators is diagrammed in Figure~\ref{fig:2schur}. Here we include the possibility of an elliptic generalization, see, e.g.,~\cite{SelbergReview}.
\begin{figure}
  \begin{center}
    \includegraphics[width=0.4\textwidth]{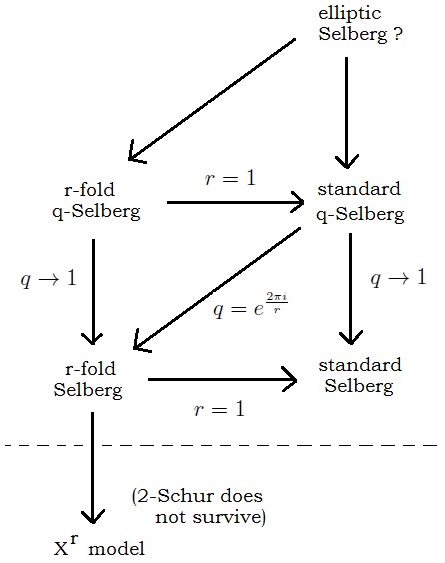}
      \caption{The landscape of two-Schur correlators\label{fig:2schur}}
  \end{center}
\end{figure}

As illustrated in Figure~\ref{fig:2schur}, the $r$-fold $q$-deformed model discussed in~\S\ref{subsec:qanalogs} generalizes all of the known solvable two-Schur correlators, with the possible exception of an elliptic version of the formula. For $r=1$ one recovers the correlators associated to the $q$-Selberg model from~\cite{Mironov:2011dk}, of which the other existing conjectures are special cases. An additional relationship in the diagram arises from the $q \to \omega_r$ limit of the $r=1$ $q$-deformed formula, which gives the $r$-fold formula without $q$-deformation, see~\S\ref{subsec:qanalogs}. If an elliptic generalization exists for $r=1$, we speculate that the $r$-fold $q$-deformed formula may be some limit of it.

In the next subsection we outline a possible physical interpretation for the existence of the factorization formula, Conjecture~\ref{conj:twoschur}.

\subsection{Five-dimensional partition functions} \label{subsec:instanton}

The AGT conjecture~\cite{Alday:2009aq,Wyllard:2009hg}, by now proven by several different approaches \cite{Dijkgraaf:2009pc,Nekrasov:2010ka,Mironov:2010pi,Alfimov:2011ju,Alba:2010qc,Yagi:2012xa,Maulik:2012wi,Tan:2013tq,Vartanov:2013ima,Aganagic:2013tta,Mironov:2015thk,Nekrasov:2015wsu,Cordova:2016cmu}, states a correspondence between conformal blocks of $q$-deformed Toda CFTs with $W$-algebra symmetry and partition functions of $\mathcal{N} = 1$ 5d gauge theories. Since conformal blocks of ($q$-deformed) Toda CFTs can be represented in matrix model form, this is really a relation between three quantities, i.e., a ``threesome'' \cite{Schiappa:2009cc} or ``triality'' \cite{Aganagic:2013tta}.

In this paper we only need the simplest example of an AGT correspondence: equality between a 3-point conformal block of $W_{q,t}( sl_2 )$---i.e., the $q$-Virasoro algebra---and the partition function of a 5d $T_2$ gauge theory. In particular, in the free-field formalism the 3-point conformal block is exactly represented by the $q$-Selberg integral. This gives a simple 5d argument for why the $q$-Selberg integral is solvable (i.e., is fully factorized as a product of Gamma-functions with linear arguments), because the partition function of a 5d $T_2$ gauge theory is fully perturbative, and does not have a non-trivial instanton part.

An important special case that has attracted significant attention is when the $q$-deformation parameter is an $r$th root of unity. In this case, as shown in \cite{Kimura:2011zf, Itoyama:2014pca}, the 5d partition function reduces to a 4d partition function on an ALE space, the ${\mathbb C}^2 / {\mathbb Z}_r$ orbifold. At the same time the matrix model ($q$-Selberg integral) reduces to precisely the integrals that we study in the current paper, see~\S\ref{subsec:qanalogs}.\footnote{It should be noted that a different matrix model was derived for the same partition function earlier in \cite{Kimura:2011zf}. This other matrix model was obtained using the methods of \cite{Sulkowski:2009br} and is not the usual AGT description. The relation between AGT-type matrix models and Sulkowski-type matrix models is not a simple change of variables; rather, it should involve a spectral duality transformation \cite{Mironov:2016cyq}.} This provides a 5d physical perspective for why the non-$q$-deformed model with $r > 1$ is solvable.

At the same time this argument also hints at a physical explanation of why the $q$-deformed model with $r > 1$ is solvable. Namely, instead of taking the root of unity limit, one can consider directly the 5-dimensional $T_2$ theory on an orbifold ${\mathbb C}^2 / {\mathbb Z}_r$. The resulting theory would depend on both $q$ and $r$ and -- because of the relation to orbifold $T_2$ theory -- should be given by a product of Gamma-functions with linear arguments. This could also shed light on why the two Schur average is factorized for these models.

To connect this answer to our integrals, however, there has to be an orbifold version of the 5d AGT correspondence. Orbifold generalizations of AGT have been considered in~\cite{Bonelli:2009zp,Belavin:2011pp,Nishioka:2011jk,Bonelli:2011jx,Bonelli:2011kv}, but we are unaware of an exact relationship to the matrix models considered in our paper. Thus, the physical interpretation of the topmost nodes in Figure~\ref{fig:2schur} remains an open problem.

This argument also allows us to make a connection to mathematics in the enumerative geometry of symplectic resolutions. The $T_2$ theory in five dimensions is known to be equivalent to a $U(1)$ theory with two fundamental hypermultiplets. The partition function of this theory is mathematically a K-theoretic integral (push-forward) of a sheaf on a Hilbert scheme of points on $\mathbb{C}^2$ that corresponds to fundamental matter. The results of the current paper correspond to considering instead a Hilbert scheme of points on $\mathbb{C}^2 / {\mathbb Z}_r$.

\subsection{Superconformal indices and surface defects}
\label{subsec:index}

Another way to explain the solvability of the $q$-Selberg integral is to note that it is a special case of the elliptic Selberg integral~\cite{vanDiejen:2001,Forrester:2007}, which is itself equal to the superconformal index of a 4d $\mathcal{N}=1$ gauge theory with gauge group $Sp(2N)$, one chiral multiplet in the antisymmetric tensor representation, and six fundamental chiral multiplets. This gauge theory is s-confining~\cite{Seiberg:1994bz,Seiberg:1994pq,Csaki:1996sm,Csaki:1996zb}, hence it admits an infrared description in terms of chiral multiplets interacting via an irrelevant superpotential. The absence of a gauge group in this description implies that the index factors into a product of elliptic gamma functions, which in turn implies the factorization of the $q$-Selberg integral.

Recall that the superconformal index is the partition function of the theory on $S^3 \times S^1$. Replacing $S^3$ by a Lens space $L(r,p)$ leads to a refined index which has received some attention in the literature~\cite{Benini:2011nc,Razamat:2013opa,Spiridonov:2016uae}. Since, of course, $L(r,n)$ is a $\mathbb{Z}_r$ orbifold of $S^3$, this suggests that perhaps the $r$-fold $q$-Selberg integral considered in~\S\ref{subsec:qanalogs} is a special case of the superconformal index of the same $Sp(2N)$ gauge theory on $L(r,n) \times S^1$ (for some $n$, e.g., $n=1$). More generally, other $r$-fold matrix models of the kind considered in~\S\ref{sec:genconstruction} may relate to the index of the same CFT on a Lens space as generates the parent matrix model on $S^3$.

An intriguing feature of this potential connection to indices is that it could explain the physics behind the factorization of Schur averages. It is natural to conjecture that the insertion of a Schur function (or the appropriate elliptic analog) into a superconformal index is related to the insertion of a surface defect in the partition function of the gauge theory, see for instance~\cite{Razamat:2013jxa}. If so, the factorization of Schur averages in the $r$-fold matrix models considered in this paper could relate to the physics of surface defects on Lens spaces.

Nonetheless, at the time of writing we have been unable to make the connection between $r$-fold matrix models and Lens space indices explicit, and it remains an interesting open problem for future research.

\section*{Acknowledgments}
We thank G.~Dunne, A.~Morozov and M.~\"Unsal for helpful discussions. C.C. is supported by a Martin and Helen Chooljian membership at the Institute for Advanced Study and DOE grant DE-SC0009988.
B.H. is supported by Perimeter Institute for Theoretical Physics, and was funded by the Fundamental Laws Initiative of the Harvard Center for the Fundamental Laws of Nature for part of this research. Research at Perimeter Institute is supported by the Government of Canada through Industry Canada and by the Province of Ontario through the Ministry of Economic Development and Innovation.  A.P. is supported in part by the NWO Vici grant. The work of S.S. is supported in part by grants RFBR 16-02-01021 and 15-31-20832-mol-a-ved.

\appendix

\section{Divisibility and Quotients of Partitions} \label{app:schurdivisibility}

In this appendix we review the concept of dividing a partition by a natural number as it relates to our work. (see, e.g., \cite[pp.\ 12--14]{Macdonald:2008zz} for a textbook treatment).

\subsection{Determinants and $p$-quotients} \label{sec:determinants}

To motivate the concepts of divisibility and quotients for partitions, consider the determinant
\be \label{eqn:exDet}
\underset{1\le I,J \le N}{\det} \delta_{p|(\lambda_I + J - I)}\, f\biggl(\frac{\lambda_I + J - I}{p}\biggr)\,, 
\ee
for natural numbers $p, N \in \mathbb{N}$, a partition $\lambda_1 \ge \ldots \ge \lambda_N \ge 0$ and any function $f(n)$. 
Writing the determinant as a sum $\det M_{I J} = \sum_{\sigma} (-1)^\sigma \prod_I M_{I \sigma(I)}$, the permutations which contribute
  satisfy
\begin{equation}
  \sigma(I) \equiv I - \lambda_I \pmod p \,.
\end{equation}
This specifies $\sigma$ up to the ordering of pairs $I, I'$ for which
$I - \lambda_I \equiv I' - \lambda_{I'} \pmod p$.
This ordering can be specified by $p$ permutations $\sigma_\mu$, $0 \leqslant \mu <
p$, constructed so that
\begin{equation}
  \sigma (I_{\mu, i}) > \sigma (I_{\mu, i'}) \qquad \mbox{iff} \qquad
  \sigma_{\mu} (i) > \sigma_{\mu} (i') \,,
\end{equation}
where $I_{\mu, i}$ is the $i$th value of $I$ for which $I-\lambda_I-1
\equiv \mu \pmod p$. We have explicitly
\begin{equation}
  \sigma (I_{\mu, i}) = 1+p (\sigma_{\mu} (i)-1) + \mu \,.
\end{equation}

Recall that a permutation $\sigma_S$ on a set $S$ is \emph{finitary} (\emph{infinitary}) if it acts non-trivially on a finite (infinite) subset of $S$. As an infinitary permutation cannot be written as a finite product of transpositions, it is convenient to define its signature as $(-1)^\sigma \equiv 0$. With this definition, it is straightforward to check that
\be
(-1)^\sigma = \delta_p(\lambda) \prod_{\mu =0}^{p-1} (-1)^{\sigma_\mu} \,,
\ee
where the ``$p$-signature'' $\delta_p(\lambda)$ is the signature of the permutation $\widehat{\sigma}(I_{\mu,i}) = 1+p (i-1) + \mu$  (c.f.\ \cite[p.\ 50]{Macdonald:2008zz}).

Since the permutations which appear in the determinant are finitary, (\ref{eqn:exDet}) vanishes unless $\delta_p(\lambda) \ne 0$. For reasons which will soon become clear, we call a partition $\lambda$ with $\delta_p(\lambda) \ne 0$ ``$p$-divisible.'' In general, we have
\be
I_{\mu, i} = 1+p (i-i_\mu-1) + \mu\,, \qquad I_{\mu, i} > \lambda^\T_1\,,
\ee
where $\widehat{\sigma}$ is finitary iff $i_\mu = 0$ for all $0 \le \mu < p$. The $p$ integers $i_\mu$, constrained by $\sum_\mu i_\mu = 0$, can be thought of as a remainder. Conventionally, this is expressed in terms of the ``$p$-core'', $\lambda \bmod p$, which is the smallest partition $\chi$ such that $i_\mu^{(\chi)} = i_\mu^{(\lambda)}$.\footnote{We later show that $\chi$ is unique, which is not obvious at present.} We use the notation $(\lambda \bmod p)_\mu \equiv i_\mu^{(\lambda)}$ henceforward.

By expressing $\sigma$ in terms of $\sigma_\mu$, the determinant~(\ref{eqn:exDet}) can be rewritten as
\be
\underset{1\le I,J \le N}{\det} \delta_{p|(\lambda_I + J - I)}\, f\biggl(\frac{\lambda_I + J - I}{p}\biggr)= \delta_p(\lambda) \prod_{\mu =0}^{p-1} \underset{1\le i,j \le n_\mu}{\det} f\biggl(j-1+\frac{\lambda_{I_{\mu,i}}- I_{\mu,i} +\mu+1}{p}\biggr) \,,
\ee
where $n_\mu \equiv \floor{\frac{N-\mu-1}{p}}+1$. Notice that
\be \label{eqn:pquotient}
(\lambda/p^{(\mu)})_i \equiv i-(\lambda \bmod p)_\mu-1 + \frac{\lambda_{I_{\mu,i}}- I_{\mu,i} +\mu+1}{p} \,,
\ee
defines a partition $\lambda/p^{(\mu)}$, since $(\lambda/p^{(\mu)})_{i+1} \le (\lambda/p^{(\mu)})_i$ and $(\lambda/p^{(\mu)})_i = 0$ for $I_{\mu,i}>N$. We find
\be \label{eqn:exDetFormula}
\boxed {\underset{1\le I,J \le N}{\det} \delta_{p|(\lambda_I + J - I)}\, f\biggl(\frac{\lambda_I + J - I}{p}\biggr)= \delta_p(\lambda) \prod_{\mu =0}^{p-1} \underset{1\le i,j \le n_\mu}{\det} f\bigl((\lambda/p^{(\mu)})_i +j- i\bigr) \,,}
\ee
since $i_\mu = 0$ when $\delta_p(\lambda) \ne 0$.

The $p$ partitions $\lambda/p^{(\mu)}$, $0\le \mu < p$, defined in~(\ref{eqn:pquotient}) are known as the ``$p$-quotient'' of $\lambda$. As we have seen, these play an important role in evaluating determinants of the form~(\ref{eqn:exDet}). More generally, the determinant of an arbitrary matrix $M_{I J}$ satisfying $M_{I J} = 0$ for $(\lambda_I + J - I) \bmod p \ne 0$ can be found using the formula
\begin{equation} \label{eqn:genDetFormula}
\boxed{
\begin{aligned}
\underset{1\le I,J \le N}{\det} \delta_{p|(\lambda_I+ J - I)} \, f_{(J-1) \bmod p}&\biggl(\floor{\frac{N+\lambda_I-I}{p}},\floor{\frac{N-J}{p}}\biggr) \\
&= \delta_p(\lambda) \prod_{\mu=0}^{p-1} \underset{1\le i,j \le n_\mu}{\det} f_\mu \bigl(n_\mu + (\lambda/p^{(\mu)})_i - i,\,n_\mu - j\bigr) \,,
\end{aligned}}
\end{equation}
for $p$ functions $f_\mu(x,y)$, where~(\ref{eqn:exDetFormula}) is a special case.

\subsection{The abacus diagram}

The $p$-core, $p$-quotient, and $p$-signature have a graphical interpretation, as follows. Consider a plot consisting of the points
\be \label{eqn:abacus1}
\biggl\{ \biggl((i-\lambda_i-1) \bmod p, \floor{\frac{\lambda_i - i}{p}} \biggr) \biggm| i\in\mathbb{N} \biggr\} \,,
\ee
in $\mathbb{Z}_p \times \mathbb{Z}$, which is the same as the plot
\be \label{eqn:abacus2}
\biggl\{ \biggl(\mu, (\lambda \bmod p)_\mu +(\lambda/p^{(\mu)})_i - i \biggr) \biggm| i \in \mathbb{N}, 0\le \mu < p\biggr\}\,.
\ee
This is known as the ``abacus diagram'', and can be constructed from the Young diagram associated to $\lambda$ as follows.

We assign a binary digit $1$ ($0$) to each vertical (horizontal) line segment along the lower-right boundary of the Young diagram. Reading off the resulting sequence from upper-right to lower-left and padding the beginning (end) with an infinite number of $0$s ($1$s), we obtain an infinite binary sequence which encodes the Young diagram, see Figure~\ref{sfig:sequence}.
\begin{figure}
  \centering
  \begin{subfigure}{0.60\textwidth}
    \centering
    \includegraphics[height=2 in]{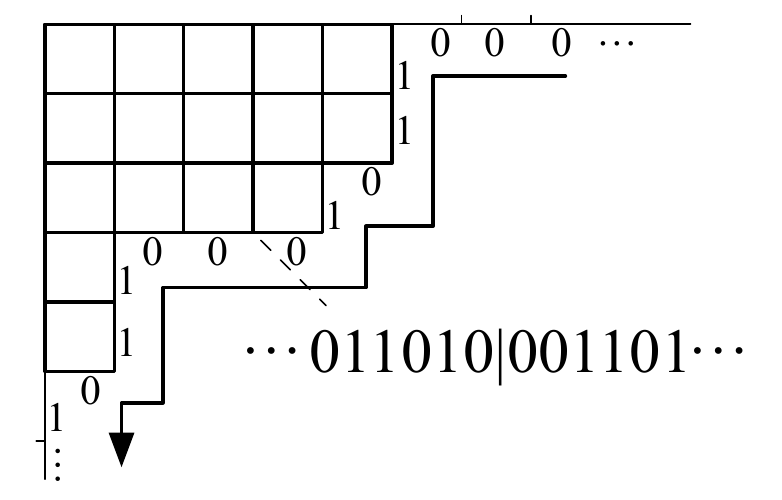}
    \caption{A Young diagram determines a binary sequence}
    \label{sfig:sequence}
  \end{subfigure}
   \begin{subfigure}{0.38\textwidth}
    \centering
    \includegraphics[height=2 in]{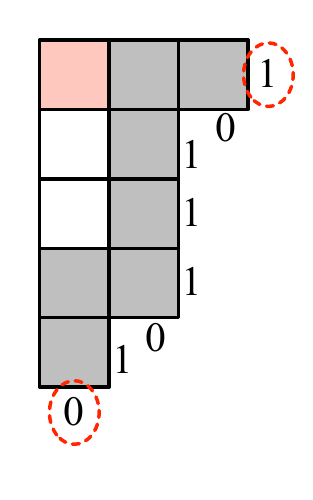}
    \caption{An inversion and its rim hook}
    \label{sfig:rimhook}
  \end{subfigure}
  \caption{\subref{sfig:sequence}~The lower-right boundary of a Young diagram determines a binary sequence, where we indicate the centerline of the sequence (the point for which the number of 1's to the left and 0's to the right are equal) with a bar. \subref{sfig:rimhook}~For each inversion in the sequence ($1$ before a $0$), there is a corresponding box in the Young diagram (red) and an associated rim hook (grey). Removing the inversion (exchanging the circled digits) corresponds to removing the rim hook.}
\end{figure}
This sequence has a natural centerline where it crosses the diagonal of the Young diagram (the point where the number of $1$s preceding is equal to the number of $0s$ following). For each ``inversion''---a pair of digits $1$, $0$, where the $1$ precedes the $0$---there is a corresponding box in the Young diagram with hook length equal to the distance between the two digits. Swapping to the two digits corresponds to removing the associated rim hook from the Young diagram, see Figure~\ref{sfig:rimhook}.

To construct the abacus diagram, we print out the binary sequence left-to-right, top-to-bottom, in $p$ columns, arranged so that the centerline falls on a line break. Each $1$ in the sequence corresponds to the end of a row in the Young diagram, whereas in the abacus diagram the $1$s coincide with the points of the plot~(\ref{eqn:abacus1}), henceforward referred to as ``beads''. Conversely, using~(\ref{eqn:abacus2}) we see that the binary sequence in the column $\mu$, $0\le \mu < p$, corresponds to the $p$-quotient $\lambda/p^{(\mu)}$, where the centerline (determined by equating the $1$s preceding and the $0$s following, as before) is a height $(\lambda \bmod p)_\mu$ above the centerline of the abacus diagram.

The remainder $(\lambda \bmod p)_\mu$ is unaffected by sliding beads vertically in the abacus diagram (adding or removing a rim-hook of length $p$), hence the $p$-core is determined uniquely by sliding all the beads downwards as far as they can go. This reproduces the standard definition of the $p$-core: the Young diagram which results from removing length $p$ rim-hooks until none remain. The relation between the Young diagram, the abacus diagram, and the $p$-quotient and $p$-core is summarized in Figure~\ref{fig:corequotient}.
\begin{figure}
  \centering
  \begin{subfigure}{0.25\textwidth}
    \centering
    \includegraphics[height=1.65 in]{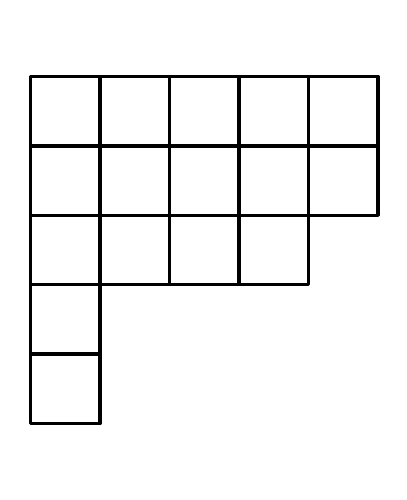}
    \caption{Young diagram}
    \label{sfig:Ydiagram}
  \end{subfigure}
   \begin{subfigure}{0.25\textwidth}
    \centering
    \includegraphics[height=1.65 in]{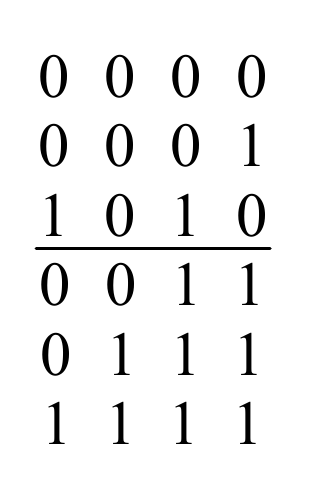}
    \caption{Abacus diagram}
    \label{sfig:abacusdiagram}
  \end{subfigure}
  \begin{subfigure}{0.45\textwidth}
    \centering
    \includegraphics[height=1.65 in]{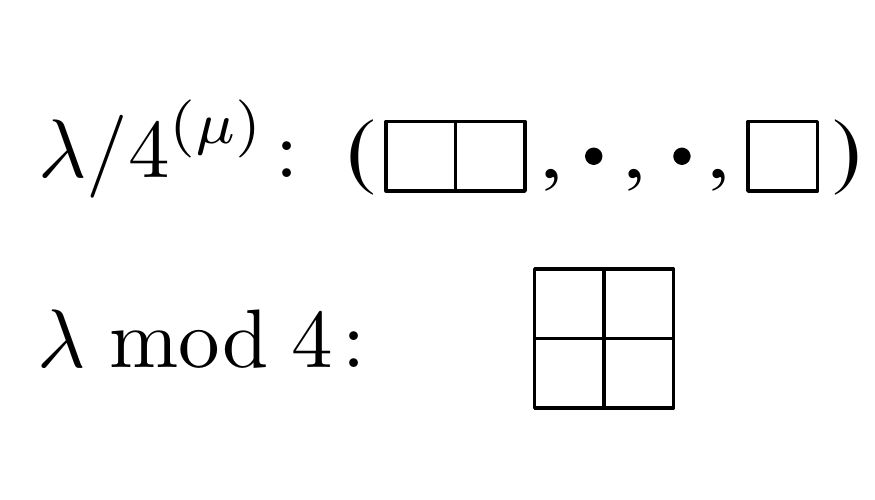}
    \caption{Core and quotient}
    \label{sfig:corequotient}
  \end{subfigure}
  \caption{Translating between~\subref{sfig:Ydiagram} a Young diagram, \subref{sfig:abacusdiagram} the corresponding abacus diagram, and~\subref{sfig:corequotient} the $p$-quotient and $p$-core, in the case $p=4$ for the Young diagram in Figure~\ref{sfig:sequence}.\label{fig:corequotient}}
\end{figure}

The $p$-signature can also be given a diagrammatic interpretation, as follows. Reading top to bottom in the $k$th column, label the $j$th bead with $k+(j-1)p$. Reading left to right, top to bottom, $\widehat{\sigma}(i)$ is the label associated to the $i$th bead, and $\delta_p(\lambda)$ is the signature of $\widehat{\sigma}$. Using this construction it is straightforward to show that the effect of adding a rim-hook spanning the rows $i, i+1, \ldots, j$ is to map
\begin{equation} \label{eqn:sigmahook}
\widehat{\sigma} \to \widehat{\sigma} \circ (j, j-1, \ldots, i)\,,
\end{equation}
in cycle notation, so that the group algebra is related to sliding beads in the abacus diagram. This observation provides another method (easier to apply by hand) for computing the $p$-signature. Removing length $p$ rim-hooks successively until no boxes remain, the $p$-signature of $\lambda$ is equal to $(-1)^{k}$ when there are $k$ rim-hooks spanning an even number of rows.

\subsection{Basic theorems}

We now catalog a few basic properties of $p$-cores, $p$-quotients, and $p$-signatures. For the sake of brevity, we omit most proofs, leaving them as an exercise for the interested reader.

To state these properties concisely, we define the functions
\begin{align}
  \mathcal{H}_{p,i}(\lambda) &= \left| \left\{ x
  \in \lambda | h_{\lambda} (x) \equiv i \pmod p \right\} \right| , & \mathcal{C}_{p,i}(\lambda) &= \left| \left\{ x
  \in \lambda | c_{\lambda} (x) \equiv i \pmod p \right\} \right| ,
\end{align}
which count the number of boxes in $\lambda$ with hook-length and contents congruent to $i \pmod p$, respectively.

We have
\begin{theorem} \label{thm:rdivisible}
  For any partition $\lambda$ and $p \in \mathbb{N}$:
  \begin{enumerate}
   \item
 $ |\lambda| = p \sum_\mu |\lambda/p^{(\mu)}| + |\lambda \bmod p|$\,.
  \item
  $\lambda$ is a $p$-core iff $\mathcal{H}_{p,0}(\lambda) = 0$.
  \item
  $\mathcal{H}_{p,0}(\lambda) \le | \lambda | / p$ and $\mathcal{H}_{p,0}(\lambda) \le \mathcal{C}_{p,0}(\lambda)$.
  \item
  The following are equivalent
  \begin{enumerate}[i.]
  \item
  $\lambda$ is $p$-divisible,
  \item
  $\mathcal{H}_{p,0}(\lambda) = | \lambda | / p$,
  \item
  $\mathcal{H}_{p,0}(\lambda) = \mathcal{C}_{p,0}(\lambda)$,
  \item
  $\forall i \in \mathbb{Z}, \mathcal{C}_{p,i}(\lambda) = |\lambda|/p$.
  \end{enumerate}
  \item
  If $\lambda$ is $p$-divisible, then $\forall i \in \mathbb{Z}$, $\mathcal{H}_{p,i}(\lambda) + \mathcal{H}_{p,p-i}(\lambda) = 2 |\lambda|/p$.
  \end{enumerate}
\end{theorem}
\noindent One way to prove these results is to construct $\lambda$ from its $p$-core by sliding beads in the abacus diagram, keeping track of both sides of the relevant (in)equality.

\begin{theorem} \label{thm:psignatureformula}
If $\lambda$ is $p$-divisible then its $p$-signature is given by the product over boxes
\be
\delta_p(\lambda) = (-1)^{\frac{|\lambda|}{p}}\prod_{x\in\lambda} (-1)^{\floor{\frac{c_\lambda(x)}{p}}+\floor{\frac{h_\lambda(x)}{p}}} \,.
\ee
\end{theorem}
\noindent As above, this formula can be proven by constructing $\lambda$ from its (empty) $p$-core and keeping track of both sides, e.g., using~(\ref{eqn:sigmahook}) to track the changes in $p$-signature.

\begin{theorem} \label{thm:transpose}
The $p$-core, $p$-quotient, and $p$-signature transform as follows under a transposition of the Young diagram
\begin{enumerate}
\item
$\lambda^\T \bmod p = (\lambda \bmod p)^\T$ \,,
\item
$\lambda^\T/p^{(\mu)} = (\lambda/p^{(p-1-\mu)})^\T$ \,,
\item
$\delta_p(\lambda^\T) = (-1)^{|\lambda|+\frac{|\lambda|}{p}} \delta_p(\lambda)$ \,.
\end{enumerate}
\end{theorem}
\noindent The first two parts are obvious from the abacus diagram, whereas the third follows from Theorem~\ref{thm:psignatureformula} as well as
\be
 \sum_{x \in \lambda} \left(\ceil{\frac{c_\lambda(x)}{p}} - \floor{\frac{c_\lambda(x)}{p}}\right) = (p-1)\frac{|\lambda|}{p} \,,
\ee
for $\lambda$ $p$-divisible, which is a direct consequence of property 4.iv of Theorem~\ref{thm:rdivisible}.

Using Theorem~\ref{thm:transpose}, we write down a useful variant of~(\ref{eqn:genDetFormula}):
\begin{equation} \label{eqn:genDetFormulaTranspose}
\boxed{
\begin{aligned}
(-1)^{|\lambda|} & \underset{1\le I,J \le K}{\det} \delta_{p|(\lambda_J^\T+ I - J)} \, f_{(-I) \bmod p}\biggl(\floor{\frac{N+I-1}{p}},\floor{\frac{N+J-\lambda_J^\T-1}{p}}\biggr) \\
&= \delta_p(\lambda) \prod_{\mu=0}^{p-1} (-1)^{|\lambda/p^{(\mu)}|} \underset{1\le i,j \le k_\mu}{\det} f_{\mu} \bigl(n_\mu + i-1,\, n_\mu +j- (\lambda/p^{(\mu)})_j^\T-1\bigr) \,.
\end{aligned}}
\end{equation}
for $K\ge \lambda_1$, where $k_\mu = \floor{\frac{K+\mu}{p}}$ and $N+K \equiv 0 \pmod p$.

\bibliography{NonGaussian}{}
\bibliographystyle{utphys}

\end{document}